%% file: MHTPD.tex
\documentclass[11pt]{article}
\usepackage{xr-hyper}
\usepackage[colorlinks,citecolor=blue]{hyperref}
\usepackage{amsfonts,amsmath,amssymb,amsthm,bm,booktabs}
\usepackage[pdftex]{graphicx}
\usepackage[hmargin=1in,vmargin=1in]{geometry}
\usepackage{setspace,enumerate,mathtools}
\usepackage{courier}
\usepackage{changepage}
\allowdisplaybreaks
\usepackage[toc,title,titletoc,header]{appendix}
\usepackage{etoolbox} 
\usepackage{multirow}
\usepackage{multicol}
\usepackage{eurosym}

\usepackage{natbib}

\newcommand{\myref}[2]{\hyperref[#1]{#2}}
\numberwithin{equation}{section}

\usepackage[nottoc,notlot,notlof]{tocbibind} 

\def\qed{\rule{2mm}{2mm}}

\usepackage[pagewise,mathlines]{lineno}
\mathchardef\dash="2D
\theoremstyle{definition}
\newtheorem{theorem}{Theorem}[section]
\newtheorem{lemma}{Lemma}[section]

\newtheorem{example}{Example}
\newtheorem{assumption}{Assumption}
\newtheorem{remark}{Remark}[section]
\AtEndEnvironment{remark}{~\qed}%

\def\qed{\rule{2mm}{2mm}}

\begin{document}
\title{Marginal homogeneity tests with panel data\thanks{We thank seminar participants at Duke Microeconometrics Breakfast for helpful comments and discussion. Of course, all errors are our own.}
}
\author{Federico A. Bugni\\Department of Economics\\Northwestern University\\ \href{mailto:federico.bugni@northwestern.edu}{\texttt{federico.bugni@northwestern.edu}}
\and
Jackson Bunting\\Department of Economics\\University of Washington\\ \href{mailto:buntingj@uw.edu}{\texttt{buntingj@uw.edu}}
\and
Muyang Ren\\Department of Economics\\University of Tennessee\\ \href{mailto:mren7@utk.edu}{\texttt{mren7@utk.edu}}}
\maketitle

\begin{abstract}
A panel dataset satisfies marginal homogeneity if the time-specific marginal distributions are homogeneous or time-invariant. Marginal homogeneity is relevant in many economic settings, including dynamic discrete games, difference-in-differences models, and finance. In this paper, we propose several tests for the hypothesis of marginal homogeneity and investigate their properties. We consider an asymptotic framework in which the number of individuals $n$ in the panel diverges, while the number of periods $T$ is fixed. We implement our tests by comparing a studentized or non-studentized $T$-sample version of the Cram\'er-von Mises statistic with a suitable critical value. We propose three methods for constructing the critical value: asymptotic approximations, the bootstrap, and time permutations. We show that the first two methods result in asymptotically exact hypothesis tests. The permutation test based on a non-studentized statistic is asymptotically exact when $T=2$, but is asymptotically invalid when $T>2$. In contrast, the permutation test based on a studentized statistic is always asymptotically exact. Finally, under a time-exchangeability assumption, the permutation test is exact in finite samples, both with and without studentization.

\begin{description}
\medskip
\item \textit{Keywords and phrases}:
Marginal homogeneity test, panel data, Cram\'er-von Mises statistic, Asymptotic approximation, Bootstrap, Permutation tests.
\medskip
\item \textit{JEL classification}: C12, C23.
\end{description}
\end{abstract}

\newpage

\section{Introduction}\label{sec:introduction}

This paper considers a hypothesis testing problem for panel data, $\{\{X_{i,t}\}_{t=1}^{T}\}_{i=1}^{n}$. We assume that the data are independent and identically distributed (i.i.d.) across units $i = 1, \dots, n$, but allow for arbitrary dependence across time. For each period $t = 1, \dots, T$, let $F_t$ denote the common marginal cumulative distribution function (CDF) of $X_{i,t}$. We say that the panel dataset satisfies {\it marginal homogeneity} if these marginal distributions are homogeneous or time-invariant, i.e.,
\begin{equation}
    F_{1}~=~F_{2}~=~\dots~=~ F_{T}.
    \label{eq:MHdefn}
\end{equation}
In this paper, we propose several tests for the hypothesis of marginal homogeneity and investigate their properties.

Marginal homogeneity is relevant across a wide range of economic settings, including dynamic discrete-choice games, difference-in-differences models, and finance. We now provide some examples of the relevance of marginal homogeneity in each of these settings.

\begin{example}[Dynamic discrete-choice games]
In dynamic discrete-choice games, researchers often observe $X_{i,t} = (A_{i,t}, S_{i,t})$, representing action and state variables for units $i = 1, \ldots, n$ (individuals, firms, etc.) over $t = 1, \ldots, T$ periods, respectively. In this context, a standard approach is to assume that the conditional choice probabilities (i.e., $P(A_{i,t} = a \mid S_{i,t} = s)$) and state transition probabilities (i.e., $P(S_{i,t+1} = s' \mid A_{i,t} = a, S_{i,t} = s)$) are homogeneous, and posit a structural model for them. Under standard assumptions, these objects yield a structural model for $P(X_{i,t+1} = x' \mid X_{i,t} = x) = f_{\theta}(x', x)$ with parameter $\theta$. This posited structure serves as the basis for inference about $\theta$ in dynamic discrete-choice games. Notably, this inference does not invoke the marginal homogeneity hypothesis in \eqref{eq:MHdefn}. However, this condition can be shown to generate efficiency gains in the estimation of $\theta$. To see this, note that marginal homogeneity in this context implies that $X_{i,t}$ is in a steady state with a marginal CDF $F$, which yields the following structural equation:
\begin{equation}
dF(x')~=~\int_{x \in \mathcal{X}} f_{\theta}(x', x) \times dF(x)~~~\text{ for all } x' \in \mathcal{X}
\label{eq:PMFs2}
\end{equation} 
Imposing \eqref{eq:PMFs2} in the estimation of $\theta$ can deliver efficiency gains relative to the standard method that does not impose it. We then can interpret the marginal homogeneity hypothesis in \eqref{eq:MHdefn} as a source of efficiency gains in the structural estimation of dynamic discrete choice games.
\end{example}

\begin{example}[Difference-in-differences models]
This example considers a difference-in-differences (DiD) model following \cite{roth2023s}. In this framework, researchers observe outcomes $Y_{i,t}$ for units $i=1,\ldots,n$ over several periods $t=1,2,\ldots,T$. Unit $i$ receives an irreversible binary treatment at period $G_i$, with $G_i=\infty$ indicating that unit $i$ is never treated. Identification of causal treatment effects in this model relies on the following parallel trends assumption: for all $t\neq t'$ and $g\neq g'$,
\begin{equation}\label{eq:parallel_trends}
\mathbb{E}\left[h(Y_{i, t}(\infty))-h(Y_{i, t^{\prime}}(\infty)) | G_i=g\right]
~=~
\mathbb{E}\left[h(Y_{i, t}(\infty))-h(Y_{i, t^{\prime}}(\infty)) | G_i=g^{\prime}\right],
\end{equation}
where $Y_{i,t}(g)$ denotes the potential outcome for unit $i$ in period $t$ if it is treated at time $g$, and $h$ is a known transformation specified by the researcher. \cite{roth/santanna:2023} note that whether the parallel trends assumption holds can depend on the choice of $h$: for instance, \eqref{eq:parallel_trends} may be valid for $h(x)=x$ but not for $h(x)=\log x$ (or vice versa). This raises the question of whether the parallel trends assumption is robust to the researcher’s choice of $h$. \citet[Proposition 3.2]{roth/santanna:2023} shows that robustness to functional form is related to marginal homogeneity of the untreated potential outcomes within each treatment group, i.e.,
\begin{equation}\label{eq:did_homo}
F_{Y_{i,1}(\infty)\mid G_i=g}=F_{Y_{i,2}(\infty)\mid G_i=g}=\dots=F_{Y_{i,T}(\infty)\mid G_i=g},
\quad \text{for each } g=1,\ldots,T,\infty.
\end{equation}
Relatedly, \citet[Remark 4]{roth/santanna:2023} note that \eqref{eq:did_homo} can also be used to validate the identifying assumption in Changes-in-Changes models \citep{athey/imbens:2006}, underscoring its empirical importance. Under the no-anticipation assumption that $Y_{i,t}(g)=Y_{i,t}(\infty)$ for all $i$ and all $t<g$ \citep[Assumption 5]{roth2023s}, one can relate \eqref{eq:did_homo} to a set of marginal homogeneity hypotheses in \eqref{eq:MHdefn}. 
In conclusion, marginal homogeneity tests for \eqref{eq:MHdefn} can be used to assess whether the parallel trends assumption is robust to the choice of functional form.
\end{example}

\begin{example}[Finance]
\cite{ditzhaus/gaigall:2022} and references therein use marginal homogeneity tests to evaluate whether two stock market indices have equal distributions of returns. Their data can be expressed as $\{\{X_{i,t}\}_{t=1}^{2}\}_{i=1}^{n}$, where $\{X_{i,1}\}_{i=1}^{n}$ denotes the monthly returns of the first index (e.g., the Nikkei 225 Stock Average) and $\{X_{i,2}\}_{i=1}^{n}$ denotes the monthly returns of the second index (e.g., the Dow Jones Industrial Average). Motivated by classical models for stock prices, \cite{ditzhaus/gaigall:2022} assume that monthly returns are i.i.d.\ across $i=1,\dots,n$ (i.e., over months), while the interconnectedness of global financial markets allows for dependence across $t=1,2$ (i.e., across indices). In this setting, the marginal homogeneity hypothesis in \eqref{eq:MHdefn} states that the two indices have identically distributed returns. Finally, we note that their analysis focuses on pairs of indices (i.e., $T=2$), whereas our framework allows for applications with $T>2$.
\end{example}

An inherent feature of the preceding examples is that the data are likely to exhibit temporal dependence. In dynamic discrete games, both actions and states can depend on their past values, which gives the problem its dynamic nature. In difference-in-difference settings, potential outcomes can depend on their lagged values. In the finance application in \cite{ditzhaus/gaigall:2022}, stock returns across the globe are likely related by the interconnectedness of global financial markets.
Beyond these examples, dependence over time is common in panel data analysis. For this reason, neither the classical two-sample literature for independent samples (e.g., \cite{darling:1957}) nor its $T$-sample generalization (e.g., \cite{kiefer:1959}) applies.

This paper studies the $T$-sample hypothesis testing problem in \eqref{eq:MHdefn} with possibly dependent data. Namely, we implement our tests by comparing a studentized or non-studentized $T$-sample version of the Cram\'er-von Mises statistic with a suitable critical value. We consider three methods for constructing the critical value: asymptotic approximations, the bootstrap, and time permutations. We show that the first two methods lead to asymptotically exact hypothesis tests, with or without studentization. Results for the permutation test are more nuanced: the permutation test based on a non-studentized statistic is shown to be asymptotically exact when $T=2$, but becomes asymptotically invalid when $T>2$. Once studentized, the permutation test is shown to be always asymptotically exact. On the other hand, relative to the non-studentized case, the asymptotic analysis of the studentized statistics requires an additional assumption: the variance-covariance matrix used in the studentization must be non-singular, an assumption that can fail in practice (see the related discussion in Section \ref{sec:AsyDist} and our empirical application in Section \ref{sec:Application}). Finally, under a time-exchangeability assumption, we show that the permutation test is exact in finite samples, both with or without studentization.

For independent cross-sectional data, the marginal homogeneity hypothesis in \eqref{eq:MHdefn} becomes the standard equality-of-distribution hypothesis for $T$-sample data, which has been thoroughly studied in the literature. In such case, \citet[Theorem 17.2.1]{lehmann/romano:2022} shows that permutation tests of homogeneity are finite-sample exact. \cite{chung/romano:2013} explores the behavior of permutation tests with studentized test statistics. 
Relatedly, \cite{bugni/horowitz:2021} studies the application of permutation tests to functional cross-sectional data. Relatively speaking, the test for the marginal homogeneity hypothesis in \eqref{eq:MHdefn} with panel data (i.e., allowing for time dependence) has received less attention. \citet{quessy/ethier:2012} study a classical Cramér-von Mises statistic and a characteristic-function-based statistic for the $T$-sample case. For the case $T=2$, \citet{wylupek:2023} proposes a combination of weighted and unweighted Kolmogorov–Smirnov statistics. Also for $T=2$, \citet{ditzhaus/gaigall:2022} extend the marginal homogeneity hypothesis to functional data, and \citet{beutner:2025} develops a modified Neyman’s smooth test that accommodates dependence across both units and time. 
However, this literature tends to exclude permutation tests, deeming them invalid without exchangeability under general dependence structure (e.g., \citet[page 750]{ditzhaus/gaigall:2022}; \citet[page 2098]{quessy/ethier:2012}), and instead relies on bootstrap methods. To our knowledge, our paper is the first to establish the theoretical validity of permutation tests of marginal homogeneity hypothesis in panel data based on time permutations. A key finding is that the non-studentized permutation test remains asymptotically valid when $T=2$, but its validity breaks down once $T>2$.

A related testing problem in dynamic discrete choice games is concerned with evaluating the homogeneity of the state transition probabilities, i.e., $P(X_{i,t+1}=x'\mid X_{i,t}=x)=P(X_{i,t'+1}=x'\mid X_{i,t'}=x)$ for all $t,t'<T-1$. See \cite{otsu/pesendorfer/takayashi:2016} and \cite{bugni2020testing} for recent contributions on this topic. The homogeneity of state transition probabilities and the marginal homogeneity in \eqref{eq:MHdefn} are non-nested hypotheses, and so our contribution is complementary to but distinct from these references.

In other related work, \cite{pauly/brunner/konietscheke:2015} and \cite{friedrich/brunner/pauly:2017} investigate the validity of permutation tests to evaluate the presence of treatment effects in experiments under factorial and repeated measure designs. There are important differences between these papers and ours. The first key distinction is in the class of data permutations used to implement the tests. \cite{pauly/brunner/konietscheke:2015} and \cite{friedrich/brunner/pauly:2017} generate their test by permuting all observations over units, treatments, and time indices. In contrast, our test relies solely on permuting the time index of the observations. We show that the classes of distributions under which the two types of permutation tests are finite-sample valid are non-nested; see Lemma \ref{lem:Pauly}. Second, our null hypotheses are different. While \citet{pauly/brunner/konietscheke:2015} and \citet{friedrich/brunner/pauly:2017} test for \textit{mean} differences of outcomes across groups defined by treatment status and time, we test for \textit{distributional} differences across time as specified in \eqref{eq:MHdefn}. Another difference is that \cite{pauly/brunner/konietscheke:2015} and \cite{friedrich/brunner/pauly:2017} focus on studentized statistics, while we consider both studentized and non-studentized statistics. In this respect, it is relevant to note that analyzing studentized statistics requires additional assumptions than analyzing non-studentized statistics; see the discussion in Section \ref{sec:AsyDist}. Finally, our Monte Carlo simulations suggest that our permutation test appears more powerful than the version based on their permutation scheme in finite samples. This is related to the fact that our permutation test only considers time index permutations, which seem to provide a better contrast to detect departures from the marginal homogeneity hypothesis in  \eqref{eq:MHdefn}. 

In another related work, \citet{konietschke/pauly:2014} considers time permutations (as their resampling scheme II) when testing mean differences in two dependent samples and shows that studentization yields an asymptotically valid inference procedure. While their analysis focuses on studentized statistics, our findings indicate that asymptotic validity can still hold for the non-studentized statistic when $T=2$. More importantly, we establish the theoretical validity of permutation tests based on the studentized statistic in settings with $T>2$, which, to the best of our knowledge, is new to the literature.

The remainder of the paper is organized as follows. Section \ref{sec:HT} introduces the hypothesis test problem in greater detail. Section \ref{sec:validity} contains our main theoretical results. Section \ref{sec:power} discusses the power of the proposed inference methods. In Section \ref{sec:MC}, we evaluate the finite-sample performance of these tests via Monte Carlo simulations. Section \ref{sec:Application} considers an empirical application based on \cite{igami/yang:2016}. Section \ref{sec:Conclusions} concludes. The paper's appendix collects all of the proofs and several auxiliary results.

\section{The hypothesis testing problem}\label{sec:HT}

This paper considers a hypothesis-testing problem for panel data with $n$ units and $T$ time periods. Inspired by the typical application in economics, we consider an asymptotic framework in which $n$ grows and $T$ remains fixed. We denote the data by $\mathbf{X}_n = \{\{X_{i,t}\}_{t=1}^{T}\}_{i=1}^{n}$. As already explained, we allow the data to be arbitrarily dependent across time $t = 1, \dots, T$, and assume i.i.d.\ across units $i = 1, \dots, n$. We formalize this assumption next.
\begin{assumption}\label{ass:A1}
For all $t=1,\dots,T$, $\{X_{i,t}\}_{i=1}^{n}$ are i.i.d.\ with marginal CDF $F_t$.
\end{assumption}
\noindent 
Our goal is to test whether the marginal homogeneity hypothesis in \eqref{eq:MHdefn} holds in the data, i.e.,
\begin{equation}
H_0~:~\eqref{eq:MHdefn}\text{ holds} ~~~~\text{vs.}~~~~H_1~:~\eqref{eq:MHdefn}\text{ does not hold}.
\label{eq:H0}
\end{equation}

We propose implementing this hypothesis test by rejecting $H_0$ in \eqref{eq:H0} whenever a test statistic exceeds a suitable critical value. That is, for any significant level of $\alpha \in (0,1)$, we propose
\begin{equation}
\phi_{n}\left( \alpha \right) ~=~1\left\{ S_{n}>{c}_{n}(1-\alpha ) \right\},\label{eq:test}
\end{equation}
where $\phi_{n}\left( \alpha \right)$ denotes the test function, $S_{n}$ the test statistic, and ${c}_{n}(1-\alpha )$ the critical value. In the remainder of this section, we describe the test statistic (Section \ref{sec:TestStat}) and establish its asymptotic distribution under the null hypothesis of marginal homogeneity (Section \ref{sec:AsyDist}). With these results in place, Section \ref{sec:validity} provides three inference methods, each based on a different type of critical value.

\subsection{Test statistics} \label{sec:TestStat}

We propose implementing our test using the Cram\'{e}r-von Mises (CvM) statistic, given by the sample-weighted sum of squared differences of the empirical CDFs for all consecutive periods. For simplicity, we evaluate these differences on a finite number of user-defined points on the real line $\mathcal{U}_K = \{u_1,u_2,\dots,u_K\}$ with $u_0:= -\infty < u_1 < u_2 < \dots < u_K$. We refer to this as the {\it non-studentized} CvM statistic, given by
\begin{align}
\label{eq:cvm}
    {S}_{n}~=~n\sum_{k=1}^{K}\sum_{t=1}^{T-1}[ \hat{F} _{t}(u_k) - \hat{F}_{t+1}(u_k) ] ^{2}\hat{P}(u_k),
\end{align}
where $\hat{F}_t$ is the empirical CDF in period $t = 1,\ldots, T$, and $\hat{P}(u_k)$ is the empirical analog of the aggregate probability in the interval $(u_{k-1},u_{k}]$, i.e., $k = 1,\ldots K$, 
\begin{equation}
\label{eq:emp_mass}
    \hat{P}(u_k)~=~ \frac{1}{nT}\sum_{t=1}^T\sum_{i=1}^n 1(u_{k-1} < X_{i,t} \leq u_k) .
\end{equation}

It is easy to see that \eqref{eq:cvm} can be reexpressed as follows
\begin{equation*}
    S_n ~=~ \hat{Z}'\hat{Z},
\end{equation*}
where $\hat{Z} \in \mathbb{R}^{(T-1)K}$ and, for all $(t,k) \in \{1,\ldots,T-1\}\times \{1,\ldots,{K}\}$,
\begin{equation}
\label{eq:Z_hat_main_tex}
 \hat{Z}_{(t-1)K+k} ~=~ \sqrt{n\hat{P}(u_k)} [\hat{F}_t(u_k) - \hat{F}_{t+1}(u_k)].
\end{equation}
Under $H_0$ in \eqref{eq:H0}, our formal arguments (see the proof of Theorem \ref{thm:asyDistSample}) reveal that $\hat{Z}$ is asymptotically distributed according to $N({\bf 0},\Sigma_Z)$, where for each $(t,k), (\tilde{t},\tilde{k}) \in \{1,\ldots,T-1\}\times \{1,\ldots,{K}\}$,
\begin{equation}
\label{eq:defnSigma}
\begin{aligned}
    &\Sigma_Z[(t-1)K+k, (\tilde{t}-1)K+\tilde{k}] ~=~ \\
    &\sqrt{P(u_k)P(u_{\tilde{k}})} \times cov[1(X_{i,t} \leq u_k) - 1(X_{i,t+1} \leq u_k), 1(X_{i,\tilde{t}} \leq u_{\tilde{k}}) - 1(X_{i,\tilde{t}+1} \leq u_{\tilde{k}})],
\end{aligned}
\end{equation}
and $P(u_k)$ is the aggregate probability in the interval $(u_{k-1},u_{k}]$,
\begin{equation}
\label{eq:pop_mass}
    P(u_k) ~=~  \frac{1}{T}\sum_{t=1}^T [F_t(u_k) - F_t(u_{k-1})].
\end{equation}
As a corollary, $S_n$ has a generalized chi-squared asymptotic distribution under $H_0$ in \eqref{eq:H0}, with weights determined by the eigenvalues of $\Sigma_Z$. Since the limiting distribution depends on $\Sigma_Z$, we refer to $S_n$ as the {\it non-studentized} CvM statistic.

If ${\Sigma}_Z$ is a non-singular matrix, it is natural also to consider a studentized version of the CvM statistic. To this end, we consider the {\it studentized} CvM statistic, given by
\begin{equation}
\label{eq:norm_cvm}
    \bar{S}_n ~=~ \hat{Z}'\hat{\Sigma}_Z^-\hat{Z},
\end{equation}
where $\hat{\Sigma}_Z^{-}$ is the generalized inverse of the empirical analog of ${\Sigma}_Z$, denoted by $\hat{\Sigma}_Z$. For each $(t,k), (\tilde{t},\tilde{k}) \in \{1,\ldots,T-1\}\times \{1,\ldots,{K}\}$, we define $\hat{\Sigma}_{Z}$ as follows:
\begin{equation}
    \begin{aligned}
    \label{eq:emp_cov}
    &\hat{\Sigma}_{Z}[(t-1)K+k, (\tilde{t}-1)K+\tilde{k}] \equiv \\
    &\sqrt{\hat{P}(u_k) \hat{P}(u_{\tilde{k}}) }\times \frac{1}{n}\sum_{i=1}^{n}
    \left[ 
    \begin{array}{l}
    \left(1( X_{i,t}\leq u_k) -\hat{F}_{t}(u_k)-1( X_{i,t+1}\leq u_k) + \hat{F}_{t+1}(u_k) \right)\times \\
    \left(1( X_{i,\tilde{t}}\leq u_{\tilde{k}}) -\hat{F}_{\tilde{t}}(u_{\tilde{k}})-1( X_{i,\tilde{t}+1}\leq u_{\tilde{k}}) + \hat{F}_{\tilde{t}+1}(u_{\tilde{k}}) \right)
    \end{array}
    \right]
    \end{aligned}
\end{equation}
for each $(t,k), (\tilde{t},\tilde{k}) \in \{1,\ldots,T-1\}\times \{1,\ldots,{K}\}$. 
If $H_0$ in \eqref{eq:H0} holds and ${\Sigma}_Z$ is a non-singular, $\bar{S}_n$ has a chi-squared asymptotic distribution with $(T-1)K$ degrees of freedom. The lack of dependence of this limiting distribution on $\Sigma_Z$ justifies referring to $\bar{S}_n $ as the {\it studentized} CvM statistic.

\begin{remark}[On the choice of test-statistics]\label{rem:ChoiceU_K}
The test statistics in \eqref{eq:cvm} and \eqref{eq:norm_cvm} are CvM-type statistics evaluated over a finite set of points $\mathcal{U}_K$ based on contrasts between consecutive empirical distributions, i.e., $\{\hat{F}_{t} - \hat{F}_{t+1}\}_{t=1}^{T-1}$. In principle, our analysis extends to other types of test statistics evaluated over these points, such as the Kolmogorov-Smirnov (KS) statistic, or contrasts relative to the average empirical CDF $\{\hat{F}_{t} - \frac{1}{T}\sum_{t=1}^T\hat{F}_t\}_{t=1}^T$, as in \cite{kiefer:1959}. It is also worth reiterating that $\mathcal{U}_K$ can be arbitrarily chosen by the researcher. In this respect, the most important simplifying feature in \eqref{eq:cvm} and \eqref{eq:norm_cvm} is that we use a finite number of points rather than an infinite number of points, such as a continuum.

There are certain reasons to prefer a finite set  $\mathcal{U}_K$ over a continuum. First and foremost, it leads to simpler asymptotic analysis regarding the studentization of test statistics. Second, for applications in which the data are discrete and with finite support $S_{X}$, one can set $\mathcal{U}_K = S_{X}\backslash\{\max S_{X}\}$ without loss of information (note that equality of marginal distributions for all points in $S_{X}\backslash\{\max S_{X}\}$ is equivalent to $H_0$). Many empirical applications, including the one in Section \ref{sec:Application}, feature discrete data with finite support. Third, studentization of empirical processes seems currently unavailable for continuous samples under general dependency. For instance, the Anderson-Darling-type standardization applied to the (truncated) KS statistic cannot resolve the issue due to the unknown dependence structure across samples \citep[Proposition 2]{wylupek:2023}. Lastly, one could in principle extend the analysis to the continuum using bootstrap prepivoting \citep{beran:1987, beran:1988}, following an approach similar to \cite{olivares:2022}. Such an extension would require new and more intricate arguments, which we view as beyond the scope of this paper.
\end{remark}

\subsection{Asymptotic distribution under the null hypothesis}\label{sec:AsyDist}

In this section, we derive the asymptotic distribution of the non-studentized CvM statistic $S_n$ in \eqref{eq:cvm} and the studentized version $\bar{S}_n$ in \eqref{eq:norm_cvm}. Our characterization of the asymptotic distribution of the studentized CvM statistic relies on the following assumption.

\begin{assumption}
\label{ass:NS}
$\Sigma_Z$ is positive definite.
\end{assumption}

We note that Assumption \ref{ass:NS} is required for our asymptotic analysis of the studentized CvM statistic and is not necessary in the case of the non-studentized version. For a suitable choice of $\mathcal{U}_K$, Assumption \ref{ass:NS} should be applicable to many data-generating processes. We now describe several scenarios in which this assumption does not hold. First, note that Assumption \ref{ass:NS} would fail if $\mathcal{U}_K$ includes any point that is ``irrelevant'' with respect to the support of ${\bf X}_n$, i.e., $\hat{P}(u_k) = 0$ for some $k = 1,\ldots K$. An example of this occurs for any $u_k \in \mathcal{U}_K$ that lies below the support of ${\bf X}_n$. Of course, one can always restore the validity of Assumption \ref{ass:NS} by removing any such points. Second and relatedly, one should never include $u_k \in \mathcal{U}_K$ that equals or exceeds the support of ${\bf X}_n$. Doing this would result in $1(X_{i,t} \leq u_k) = 1(X_{i,t+1} \leq u_k) = 1$, leading to $1(X_{i,t} \leq u_k) - 1(X_{i,t+1} \leq u_k) = 0$, rendering $\Sigma_Z$ singular.
Finally, $\Sigma_Z$ would be singular if there is no communication between some states. For example, consider a Markov chain for $X_{i,t}\in \{1,2,3,4\}$ with $P( X_{i,t+1} = x' | X_{i,t} = x)= \Pi[x,x']$ for all $x,x'\in \{1,2,3,4\}$, where for any $\rho_1,\rho_2,\rho_3,\rho_4 \in (0,1)$,
\begin{equation*}
       \Pi  ~=~ 
        \begin{bmatrix}
            \rho_1 & 1-\rho_1 & 0   & 0 \\
            1-\rho_2 & \rho_2 & 0   & 0 \\
            0   & 0   & \rho_3 & 1-\rho_3 \\
            0   & 0   &  1-\rho_4& \rho_4
        \end{bmatrix}
\end{equation*}
This transition matrix implies no communication between the first and last two states. Then, $1(X_{i,t} \leq u_k) - 1(X_{i,t+1} \leq u_k) = 0$ for any $u_k \in [2,3)$, producing a singular $\Sigma_Z$.

\begin{remark}
    One possible way to relax Assumption \ref{ass:NS} is to apply bootstrap prepivoting  \citep{beran:1987, beran:1988} directly to the non-studentized statistic $S_n$, rather than studentizing it using the joint sample covariance matrix $\hat{\Sigma}_Z$. In this case, it may suffice to require only that the limiting distribution $S$ be nondegenerate. This approach has proven effective in other permutation test settings, such as tests for equality of parameters across independent samples (see, e.g., \citet[Section 2.2]{chung/romano:2016} and \citet{fogarty:2021}). We consider this out of the scope of our work.
\end{remark}

The following result establishes the asymptotic distribution of the non-studentized and studentized CvM statistics under the marginal homogeneity hypothesis in \eqref{eq:H0}.

\begin{theorem} \label{thm:asyDistSample}
Let Assumption \ref{ass:A1} hold.
\begin{enumerate}
    \item[(a)] Under $H_0$ in \eqref{eq:H0}, 
    \begin{equation}
    S_n~~\overset{d}{\to }~~S~\equiv~\sum_{j=1}^{(T-1)K}\lambda _{j}\zeta _{j}^{2}, \label{eq:LimitingS}
    \end{equation}
    where $\{\lambda _{j}\}_{j=1}^{(T-1)K}$ are the eigenvalues of $\Sigma_Z$ in \eqref{eq:defnSigma} and $\{\zeta_j\}_{j=1}^{(T-1)K}$ are i.i.d.\ $N(0,1)$. 
    \item[(b)] Under Assumption \ref{ass:NS} and $H_0$ in \eqref{eq:H0}, 
    \begin{equation}
        \bar{S}_n ~~\overset{d}{\to }~~\chi_{(T-1)K}^2, \label{eq:LimitingS_norm}
    \end{equation}
    where $\chi_{(T-1)K}^2$ denotes the chi-squared distribution with $(T-1)K$ degrees of freedom. 
\end{enumerate}
\end{theorem}

Theorem \ref{thm:asyDistSample} shows that under the marginal homogeneity hypothesis, the non-studentized CvM statistic in \eqref{eq:cvm} converges to a generalized chi-square distribution with the weights determined by the eigenvalues of $\Sigma_Z$. Since $\Sigma_Z$ is not necessarily positive definite, some eigenvalues may be zero, leading to reduced degrees of freedom. When $\Sigma_Z$ is positive definite, then the limiting distribution of the studentized CvM statistic in \eqref{eq:norm_cvm} is chi-square distributed. These results are the basis of the critical values proposed in the following section.

\section{Critical values and validity of inference}\label{sec:validity}

This section describes three critical values for the CvM test statistics proposed in Section \ref{sec:TestStat}. Each critical value generates a different hypothesis test according to equation \eqref{eq:test}. We formally study the validity of each one of these methods.

\subsection{Asymptotic approximation} \label{sec:AA}

In this section, we propose a hypothesis test for \eqref{eq:H0} by approximating the quantiles of the asymptotic distribution in Theorem \ref{thm:asyDistSample}. To this end, we now introduce some notation. For any $\ell\in \mathbb{R}^{(T-1)K}$, let $S(\ell)\geq 0$ denote a random variable with the generalized chi-square distribution of weights equal to $\ell$, unit vector of degrees of freedom, zero vector of non-centrality parameters, and no constant or normal terms. Also, for any $(x,\ell)\in \mathbb{R} \times \mathbb{R}^{(T-1)K}$, let $G(x,\ell)$ denote the CDF of $S(\ell) $ evaluated at $x\in\mathbb{R}$. This function can be numerically computed with arbitrary accuracy by simulating its empirical distribution.

For the non-studentized CvM statistic, we propose
\begin{equation}
{c}_{n}^{A}( 1-\alpha ) ~=~\inf \big\{ x\in \mathbb{R} :G(x,\hat{\lambda}_{n})\geq 1-\alpha \big\},  \label{eq:cnAA}
\end{equation}
where $\hat{\lambda}_n$ denotes the eigenvalues of $\hat{\Sigma}_Z$, and the following hypothesis test:
\begin{equation}
\phi_{n}^{A}( \alpha ) ~=~1\left\{ S_{n}>{c}_{n}^{A}(1-\alpha ) \right\}.  \label{eq:testAA}
\end{equation}
For the studentized CvM statistic, we propose the following hypothesis test:
\begin{equation}
\bar\phi_{n}^{A}( \alpha ) ~=~1\left\{ \bar{S}_{n}>\bar{c}_{n}^{A}(1-\alpha ) \right\},   \label{eq:testAA_norm}
\end{equation}
where $\bar{c}_{n}^{A}(1-\alpha )$ equals the $(1-\alpha)$-quantile of the (standard) chi-squared distribution with $(T-1)K$ degrees of freedom.

The following result shows that the hypothesis tests in \eqref{eq:testAA} and \eqref{eq:testAA_norm} are asymptotically valid.

\begin{theorem}\label{thm:AA_fixed}
Let Assumption \ref{ass:A1} and $H_{0}$ in \eqref{eq:H0} hold, and let $\alpha \in (0,1)$.
\begin{enumerate}
    \item[(a)] $\lim_{n\to\infty} E_{P}[ \phi _{n}^{A}(\alpha)] \leq \alpha$. Furthermore, the inequality becomes an equality if $\Sigma_Z \neq 0_{(T-1)K \times (T-1)K}$.
    \item[(b)] Under Assumption \ref{ass:NS}, $\lim_{n\to\infty} E_P[\bar{\phi}^A_n(\alpha)] = \alpha$.
\end{enumerate}
\end{theorem}

\subsection{Bootstrap} \label{sec:BS}

This section proposes a hypothesis test for \eqref{eq:H0} via the bootstrap. To this end, we repeatedly resample the data with replacement across units $i=1,\dots,n$ to construct a bootstrap sample, denoted by $\{X_i^*\}_{i=1}^n$. For each bootstrap sample, the bootstrap analogs of the non-studentized and studentized CvM statistics are given by
\begin{equation}
\label{eq:bootStat}
    S_n^* ~=~ (\hat{Z}^*)'(\hat{Z}^*)~~~~\text{and}~~~~\bar{S}_n^* ~=~ (\hat{Z}^*)'\hat{\Sigma}_{Z}^-(\hat{Z}^*),
\end{equation}
where, for all $(t,k) \in \{1,\ldots, T-1\} \times \{1,\ldots, K\}$, 
\begin{equation}
        \hat{Z}^*_{(t-1)K + k} ~\equiv~ \sqrt{\hat{P}(u_k)} \frac{1}{\sqrt{n}}\sum_{i=1}^n \big(\big[ 1(X_{i,t}^* \leq u_k) - \hat{F}_t(u_k) \big] - \big[ 1(X_{i,t+1}^* \leq u_k) - \hat{F}_{t+1}(u_k) \big]\big).
      \label{eq:Z_hat_star}
\end{equation}

\begin{remark}\label{rem:bootEstVar}
    One could also define $\bar{S}_n^*$ in \eqref{eq:bootStat} with $\hat{\Sigma}_{Z}^-$ replaced by its bootstrap analog. Our main text omits this option for brevity, but we include it in our Monte Carlo simulations.
\end{remark}

By repeating the bootstrap sampling sufficiently many times, we can approximate the conditional distributions $P( S_n^*\leq x|\mathbf{X}_n) $ and $P( \bar{S}_n^*\leq x|\mathbf{X}_n) $ with arbitrary accuracy.

For the non-studentized CvM statistic, we propose
\begin{equation}
{c}_{n}^{B}( 1-\alpha ) ~=~\inf \left\{ x\in \mathbb{R} :P( S_{n}^{\ast }\leq x|\mathbf{X}_n)\geq 1-\alpha \right\},  \label{eq:cn_B}
\end{equation}
and the following hypothesis test:
\begin{equation}
\phi_{n}^{B}( \alpha ) ~=~1\left\{ S_{n}>{c}_{n}^{B}(1-\alpha ) \right\}.  \label{eq:testB}
\end{equation}
For the studentized CvM statistic, we propose
\begin{equation}
\bar{c}_{n}^{B}( 1-\alpha ) ~=~\inf \left\{ x\in \mathbb{R} :P( \bar{S}_{n}^{\ast }\leq x|\mathbf{X}_n)\geq 1-\alpha \right\},  \label{eq:cn_B_norm}
\end{equation}
and the following hypothesis test:
\begin{equation}
\bar{\phi}_{n}^{B}( \alpha ) ~=~1\left\{\bar{S}_{n}>\bar{c}_{n}^{B}(1-\alpha ) \right\}.  \label{eq:testB_norm}
\end{equation}

The next result shows that the hypothesis tests in \eqref{eq:cn_B} and \eqref{eq:cn_B_norm} are asymptotically valid.

\begin{theorem}\label{thm:B_fixed} 
Let Assumption \ref{ass:A1} and $H_{0}$ in \eqref{eq:H0} hold, and let $\alpha \in (0,1)$.
\begin{enumerate}
    \item[(a)] $\lim_{n\to\infty} E_{P}[ \phi _{n}^{B}(\alpha)] \leq \alpha$. Furthermore, the inequality becomes an equality if $\Sigma_Z \neq 0_{(T-1)K \times (T-1)K}$.
    \item[(b)] Under Assumption \ref{ass:NS},  $\lim_{n\to\infty} E_P[\bar{\phi}_n^B(\alpha)] = \alpha$.
\end{enumerate}
\end{theorem}

\subsection{Permutations} \label{sec:PT}

In this section, we propose a hypothesis test for \eqref{eq:H0} by random permutations of the data. Our permutations are motivated by the marginal homogeneity hypothesis in \eqref{eq:MHdefn}, and they consist of randomly permuting the time index $t=1,\dots, T$ for each unit $i=1,\dots,n$.

These tests require the following notation. Let $\mathcal{M}$ denote the set of all permutations of the indices $\{ 1,\ldots,T\} $, and $\mathcal{M}^{n}$ is defined as the set of all possible permutations of the time indices over $n$ observations. A typical element of $\mathcal{M}^{n}$ is given by $\pi^n  \equiv \{ \{ \pi _{i}( 1) ,\ldots ,\pi _{i}( T) \} \} _{i=1}^{n}$, where $\pi _{i}(t)$ denotes the permuted time index of the observation $X_{i,t}$, and $\{\pi_i(1), \ldots, \pi_i(T)\}$ is an arbitrary time permutation that belongs to the set $\mathcal{M}$. In other words, the permuted version of the data $\mathbf{X}_n = \{\{X_{i,t}\}_{t=1}^{T}\}_{i=1}^{n}$ can be written as $\mathbf{X}_n^{\pi } \equiv \{\{X_{i,\pi _{i}(t)}\}_{t=1}^{T}\}_{i=1}^{n}$.

For each permutation $\pi^n  \in \mathcal{M}^{n}$, the permutation analog of the non-studentized and studentized CvM statistics are given by
\begin{equation}
    S_n^\pi = (\hat{Z}^\pi)'(\hat{Z}^\pi)     ~~~~\text{and}~~~~   \bar{S}_{n}^{\pi } = (\hat{Z}^\pi)' \hat{\Sigma}^{-}_{Z^\pi} (\hat{Z}^\pi) , 
\label{eq:permStat}
\end{equation}
where, for all $(t,k) \in \{1,\ldots,T-1\}\times \{1,\ldots,K\}$,
\begin{equation}
\label{eq:Z_hat_perm}
    \hat{Z}^\pi_{(t-1)K+k} \equiv \sqrt{n\hat{P}^\pi(x)} \big(\hat{F}^\pi_t(u_k) - \hat{F}^\pi_{t+1}(u_k)\big),
\end{equation}
and, for all $(t,k), (\tilde{t}, \tilde{k}) \in \{1,\ldots,T-1\}\times \{1,\ldots,K\}$,
\begin{equation}
    \begin{aligned}
    \label{eq:perm_emp_cov}
    &\hat{\Sigma}_{Z^\pi}[(t-1)K+k, (\tilde{t}-1)K+\tilde{k}] \equiv \\
    &\sqrt{\hat{P}^\pi(u_k) \hat{P}^\pi(u_{\tilde{k}}) }\times \frac{1}{n}\sum_{i=1}^{n}
    \left[ 
    \begin{array}{l}
    \left(1( X_{i,\pi_i(t)}\leq u_k) -\hat{F}^\pi_{t}(u_k)-1( X_{i,\pi_i(t+1)}\leq u_k) + \hat{F}^\pi_{t+1}(u_k) \right)\times \\
    \left(1( X_{i,\pi_i(\tilde{t})}\leq u_{\tilde{k}}) -\hat{F}^\pi_{\tilde{t}}(u_{\tilde{k}})-1( X_{i,\pi_i(\tilde{t}+1)}\leq u_{\tilde{k}}) + \hat{F}^\pi_{\tilde{t}+1}(u_{\tilde{k}}) \right)
    \end{array}
    \right],
    \end{aligned}
\end{equation}
with $\hat{P}^{\pi }( u_k) \equiv\frac{1}{nT}\sum_{t=1}^{T}\sum_{i=1}^{n}1( u_{k-1} < X_{i, \pi_i ( t)} \leq u_k) $ and $\hat{F}^\pi_t(u_k) \equiv \frac{1}{n} \sum_{i=1}^n 1(X_{i,\pi_i(t)} \leq u_k)$.

These permutation test statistics can be used to construct permutation-based tests along the lines of \citet[Section 17.2.1]{lehmann/romano:2022}. For the non-studentized CvM statistic, we propose a critical value ${c}_n^\pi(1-\alpha)$, which is the $(1-\alpha)$-quantile over all possible permutations of the non-studentized CvM statistics (denoted by $\left\{ {S}_{n}^{\pi } : \pi^n \in \mathcal{M}^{n} \right\}$). The corresponding hypothesis test is defined as:
\begin{equation}
\phi_{n}^{\pi}( \alpha ) ~=~1\left\{ S_{n}>{c}_{n}^{\pi}(1-\alpha ) \right\}.  \label{eq:test_perm}
\end{equation}
For the studentized CvM statistic, we propose an analogous object. That is, $\bar{c}_n^\pi(1-\alpha)$ equals to the $(1-\alpha)$-quantile over all the possible permutations for the studentized CvM statistics (given by $\{ \bar{S}_{n}^{\pi }: \pi^n \in \mathcal{M}^{n}\}$).
The corresponding hypothesis test is given by:
\begin{equation}
\bar\phi_{n}^{\pi}( \alpha ) ~=~1\left\{ \bar{S}_{n}>\bar{c}_{n}^{\pi}(1-\alpha ) \right\}.  \label{eq:test_perm_norm}
\end{equation}

\begin{remark}\label{rem:nonRandomPerm}
The tests in \eqref{eq:test_perm} and \eqref{eq:test_perm_norm} are the non-random versions of the standard permutation tests described in \citet[Section 17.2.1]{lehmann/romano:2022}. The key difference between the non-random and random versions lies in how ties between the test statistic and the critical value are handled: the non-random version does not reject, while the random version rejects with a specific probability. While more conservative in finite samples, the non-random version is preferred because it is non-stochastic, simpler, and yields the same asymptotic behavior as its random counterpart for the tests under consideration.
\end{remark}

The next result studies the asymptotic validity of the hypothesis tests in \eqref{eq:test_perm} and \eqref{eq:test_perm_norm}.

\begin{theorem}\label{thm:P_fixed} 
Let Assumption \ref{ass:A1} and $H_{0}$ in \eqref{eq:H0} hold, and let $\alpha \in (0,1)$.
\begin{enumerate}
    \item[(a)] For $T=2$, $\lim_{n\to\infty} E_{P}[ \phi _{n}^{\pi}(\alpha)] \leq \alpha$. Furthermore, the inequality becomes an equality if $\Sigma_Z \neq 0_{(T-1)K \times (T-1)K}$.
    \item[(b)] When $T>2$, the non-studentized permutation test in \eqref{eq:test_perm} is invalid. That is, it is possible to have $\liminf_{n\to\infty} E_P[\phi _{n}^{\pi}(\alpha)] > \alpha$ for some distribution $P$ that satisfies our assumptions.
    \item[(c)] Under Assumption \ref{ass:NS},  $\lim_{n\to\infty} E_P[\bar\phi_{n}^{\pi}(\alpha)] = \alpha$.
\end{enumerate}
\end{theorem}

We first describe the result for the non-studentized test in \eqref{eq:test_perm}. Part (a) shows that the non-studentized test is asymptotically valid for $T=2$.  Part (b) reveals that the previous result fails for $T>2$. This finding is apparently new in the literature and is empirically relevant, as many economic applications involve panel data with more than two periods. To gain intuition about (a) and (b), it is useful to compare the variance-covariance of the $i$th observation, i.e., $\{X_{i,t}\}_{t=1}^{T}$, before and after permutations. If the variance-covariance matrix remains the same, the permutation test can be shown to be asymptotically valid. See the proof of Theorem \ref{thm:P_fixed} for a justification. When $T=2$, the $i$th observation is $(X_{i,1},X_{i,2})$ before the permutation and a mixture of $(X_{i,1},X_{i,2})$ and $(X_{i,2},X_{i,1})$ after the permutation. Under $H_{0}$ in \eqref{eq:H0}, these two random vectors share the same variance-covariance matrix. When $T>2$, this equivalence breaks down. For example, when $T=3$, the $i$th observation is $(X_{i,1},X_{i,2},X_{i,3})$ before the permutation and a mixture of $(X_{i,1},X_{i,2},X_{i,3})$, $(X_{i,1},X_{i,3},X_{i,2})$, $(X_{i,2},X_{i,1},X_{i,3})$, $(X_{i,2},X_{i,3},X_{i,1})$, $(X_{i,3},X_{i,1},X_{i,2})$, and  $(X_{i,3},X_{i,2},X_{i,1})$ after the permutation. The variance-covariance matrix of the former and the latter can differ even under $H_{0}$ in \eqref{eq:H0}, which renders the permutation test invalid.

Finally, part (c) of Theorem \ref{thm:P_fixed} shows that the studentized test in \eqref{eq:test_perm_norm} is asymptotically valid for any $T \geq 2$. This result is in line with several studies in the literature that prove the asymptotic validity of permutation tests for suitably normalized test statistics, e.g., \cite{janssen:1997,chung/romano:2013,diciccio/romano:2017,chung/olivares:2021}. This result is also related to the earlier discussion regarding variance-covariance matrices of the $i$th observation before and after permutations. After studentization, the variance-covariance matrix of the $i$th observation is asymptotically invariant to permutations. By the logic of the previous paragraph, this implies that the permutation test is asymptotically valid for the studentized test statistic.

An important advantage of the permutation tests over the ones described in previous subsections lies in their finite-sample validity under an important class of distributions that satisfy $H_{0}$ in \eqref{eq:H0}. To explain this clearly, let $\Omega_{\rm TE}$ denote the set of distributions that satisfy time exchangeability, i.e., for any $i=1,\dots,n$, $\{X_{i,t}\}_{t=1}^T $ has the same distribution as $\{X_{i,\pi(t)}\}_{t=1}^T$ for each $\pi \in \mathcal{M}$. The following result describes the finite-sample validity of our test under suitable conditions.

\begin{theorem}
\label{thm:fsperm}
Let Assumption \ref{ass:A1} hold. Then,
\begin{enumerate}
\item[(a)] $P \in \Omega_{\rm TE}$ implies that $P$ satisfies $H_{0}$ in \eqref{eq:H0}. However, the converse does not hold.
\item[(b)] For each $P \in \Omega_{\rm TE}$, both of our permutation tests are finite-sample valid, i.e.,
\begin{equation}
            E_P[\phi_n^\pi(\alpha)] \leq \alpha 
       ~~~~ \text{and} ~~~~
        E_P[\bar{\phi}_n^\pi(\alpha)] \leq  \alpha.
        \label{eq:FS_SizeControl}
\end{equation}
\end{enumerate}
\end{theorem}

Theorem \ref{thm:fsperm} implies that our permutation tests are finite-sample valid under suitable conditions. Part (a) says that a time-exchangeable distribution satisfies the null hypothesis $H_{0}$ in \eqref{eq:H0} under our maintained Assumption \ref{ass:A1}. Under such distribution, the permutation tests provide finite-sample size control. We stress that size control is not exact (i.e., the inequalities in \eqref{eq:FS_SizeControl} might be strict) only because we are using a non-randomized permutation test; see Remark \ref{rem:nonRandomPerm}. That is, if we replaced our permutation tests with their random versions, these would enjoy exact size control (i.e., both inequalities in \eqref{eq:FS_SizeControl} would hold with equality).

The combination of Theorems \ref{thm:P_fixed} and \ref{thm:fsperm} justifies the use of the studentized permutation test in \eqref{eq:test_perm_norm} (and also the non-studentized one in \eqref{eq:test_perm} when $T=2$). Theorem \ref{thm:P_fixed} indicates that this test is asymptotically exact, and Theorem \ref{thm:fsperm} shows that it is finite-sample valid for an important class of distributions in $\Omega_{\rm TE}$. As already mentioned, the finite-sample validity makes the permutation test an exceptionally attractive inference method compared to those discussed in previous subsections.

\begin{remark}\label{rem:pauly}
It is worth noting that our class of time permutations differs from the class of ``all permutations'' that uniformly permute both time periods and units. The latter has been used in the previous literature, such as \cite{friedrich/brunner/pauly:2017}. Naturally, our time permutations form a strict subset of the class of ``all permutations''. On the flip side, under our maintained Assumption \ref{ass:A1}, the class of distributions that satisfy exchangeability over ``all permutations'' is {\it more restrictive} than the set of distributions that satisfy time exchangeability only; see Lemma \ref{lem:Pauly}. In other words, if we used a permutation test based on ``all permutations'', we could only establish Theorem \ref{thm:fsperm} for a substantially smaller class of distributions than $\Omega_{\rm TE}$. Importantly, however, our objective is not to test whether $P \in \Omega_{\rm TE}$ or whether the distribution is fully exchangeable, but to test the null hypothesis \eqref{eq:H0} on equality of distributions across periods. 
\end{remark}

\section{Power analysis}\label{sec:power}

This section briefly describes the power properties of the various marginal homogeneity tests considered in this paper. Given the results in Section \ref{sec:validity}, we restrict attention to the hypothesis tests that are asymptotically valid under Assumptions \ref{ass:A1}–\ref{ass:NS}. These are as follows:
\begin{itemize}
    \item The asymptotic approximation-based test, both non-studentized and studentized.
    \item The bootstrap-based test, both non-studentized and studentized.
    \item The studentized permutation-based test.
    \item The non-studentized permutation-based test for $T=2$.
\end{itemize}

As explained in Section \ref{sec:TestStat}, our CvM test statistic is defined to detect differences in marginal CDFs at any point in $\mathcal{U}_{K}$. We thus focus our power analysis on the following subset of $H_1$: 
\begin{equation}
    \bar{H}_{1}: F_t(u)\neq F_{r}(u)~~\text{ for some }~t,r=1,\dots,T~\text{ and }~u \in \mathcal{U}_{K}.
\end{equation}
For all fixed hypotheses in $\bar{H}_{1}$, it is not hard to see that the CvM test statistics in \eqref{eq:cvm} and \eqref{eq:norm_cvm} diverge. At the same time, one can establish that the critical values described throughout this paper remain bounded in probability. For this reason, all asymptotically valid tests are consistent against any fixed hypothesis in $\bar{H}_{1}$.

We now compare the local power properties of these tests. We consider local alternative hypotheses under $\bar{H}_{1}$, which are sequences of DGPs whose marginal CDFs $\{F_{n,t}:t=1,\dots,T\}$ satisfy $\{\sqrt{n} (F_{n,t}(u_k) - F_{n,r}(u_k)):t,r=1,\dots,T,u \in \mathcal{U}_{K}\} \to c \in \mathbb{R}^{T^2 K}$ with $c'c > 0$. Under these sequences of distributions, we can repeat the arguments used to prove Theorem \ref{thm:asyDistSample} to establish the asymptotic distribution of the CvM test statistic in \eqref{eq:cvm} and \eqref{eq:norm_cvm}. It is not hard to see that these become the non-central versions of the asymptotic distributions in \eqref{eq:LimitingS} and \eqref{eq:LimitingS_norm} under $H_0$. Moreover, under these local alternatives, the critical values used throughout the paper can be shown to have the same asymptotic behavior. As a corollary, the asymptotically valid tests based on the non-studentized statistics share the same local power properties, and the same holds for those based on the studentized statistics.

\section{Monte Carlo simulations}\label{sec:MC}

This section investigates the finite-sample performance of the marginal homogeneity tests proposed in this paper. To this end, we repeatedly simulate independent panel datasets ${\bf{X}}_n = \{\{X_{i,t}\}_{t=1}^{T}\}_{i=1}^{n}$ where, for each $i=1,\dots,n$ and $t=1,\dots,T$,
\[
    X_{i,t} ~=~ (\varepsilon_{i,t} + \xi_i)\alpha_t
\]
with
\[
    \left(
    \begin{array}{c}
       \{\varepsilon_{i,t}\}_{t=1}^T \\
     \xi_i
    \end{array}
    \right)~\sim~ 
    N\left(0_{(T+1)\times 1}, \begin{bmatrix}
    \Xi & 0_{T\times 1} \\
    0_{1\times T} & 1
    \end{bmatrix}\right),
\]
where $\{\alpha_t\}_{t=1}^{T}$ is a sequence of constants and $\Xi$ is a constant positive-definite matrix. By definition, $\xi_i$ is a random effect and $\{\varepsilon_{i,t}\}_{t=1}^T$ are transient shocks with variance-covariance matrix $\Xi$. We consider two specifications for $\Xi$:
\begin{itemize}
\item WN: $\Xi = I_{T\times T}$, i.e., $\{\varepsilon_{i,t}\}_{t=1}^T$ is a white noise process with zero mean and unit variance. The distribution of ${\bf{X}}_n$ is time-exchangeable  whenever $\alpha_t$ is constant across $t$.
\item AR(1): for all $t_1,t_2=1,\dots,T$ with $t_1\neq t_2$, $\Xi[t_1,t_1]=1$ and $\Xi[t_1,t_2]=(-0.9)^{|t_1-t_2|}$, i.e., $\{\varepsilon_{i,t}\}_{t=1}^T$ is an AR(1) process zero mean, variance one, and correlation coefficient $-0.9$. In this case, the distribution of ${\bf{X}}_n$ is not time-exchangeable.
\end{itemize}
We consider two options for $\{\alpha_t\}_{t=1}^{T}$, which determine whether the data satisfies marginal homogeneity or not. For simulations under $H_0$ in \eqref{eq:H0}, we use $\alpha_t = 1$ for $t=1,\dots,T$. For simulations under $H_1$ in \eqref{eq:H0}, we set $\alpha_t = \sqrt{1 + (t-1)/2}$ for $t=1,\dots,T$.

To compute the CvM statistics we set $\mathcal{U}_K$ equal to the 1/6, 2/6, 3/6, 4/6, and 5/6 empirical quantiles of $\{\{X_{i,t}\}_{t=1}^{T}\}_{i=1}^{n}$ (thus, $K=5$). For each dataset, we implemented the following tests with a significance level of $\alpha = 5\%$:
\begin{itemize}
    \item AA: Asymptotic approximation tests in Section \ref{sec:AA}, non-studentized as in \eqref{eq:testAA} and studentized by sample covariance matrix $\hat{\Sigma}_Z$ as in \eqref{eq:testAA_norm}.
    \item BS: Bootstrap tests in Section \ref{sec:BS}, non-studentized as in \eqref{eq:testB} and studentized as in \eqref{eq:testB_norm}, the latter using the sample covariance matrix $\hat{\Sigma}_Z$ for studentization.
    \item BS2: Studentized bootstrap tests based on a bootstrapped covariance estimator, as described in Remark \ref{rem:bootEstVar}.
    \item PT: Permutation tests in Section \ref{sec:PT}, non-studentized as in \eqref{eq:test_perm} and studentized as in \eqref{eq:test_perm_norm}, the latter using the permutation covariance matrix $\hat{\Sigma}_Z^\pi$ for studentization.
    \item PT2: Permutation tests based on the class of ``all permutations'' (i.e., time periods and units), as described in Remark \ref{rem:pauly}. These can be non-studentized and studentized.
\end{itemize}

We conduct simulations with $n \in \{30,60,120,240,480\}$ units and $T \in \{2,3\}$ periods. The results shown in the tables are obtained from $S=5,000$ independent panel data draws from the design under consideration.

Table \ref{tab:panel_size} describes the rejection rates of the various tests under marginal homogeneity, i.e., $H_0$ in \eqref{eq:H0}. When $T=2$, all our proposed hypothesis tests (AA, BS, PT, studentized or not) provide exact size control when $n$ is sufficiently large. These results are consistent with our asymptotic analysis. This conclusion also seems to apply to the BS2 and the studentized PT2 tests. On the other hand, the non-studentized PT2 test does not control size. This is consistent with our discussion in Lemma \ref{lem:Pauly}, where we argue that the permutation class used to implement our PT tests is qualitatively different from the one used for the PT2 tests. Our simulations also allow us to examine how our asymptotically valid methods perform when $n$ is relatively small. In this respect, one interesting finding is that the non-studentized versions of the AA and BS tests outperform their corresponding studentized ones when $n$ is small. On the other hand, the PT test performs equally well with and without studentization for small $n$. The results for $T=3$ are qualitatively similar to those for $T=2$, except for the PT test. For $T>2$, our formal results show that the studentized PT test is asymptotically valid, but the non-studentized PT test is not. This is clearly evidenced in our simulations with AR(1) shocks, where the non-studentized PT test exhibits rejection rates close to 13\%. However, if the data are time-exchangeable, as with the WN shocks, we observe that the PT test is finite-sample valid regardless of studentization, even when $T=3$.

\begin{table}[htbp]
\centering
\scalebox{1}{
\begin{tabular}{ccccrrrrr}
\hline \hline
$T$     &  $\Xi$ & Test type      & Critical value & $n=30$  & $n=60$  & $n=120$ & $n=240$ & $n=480$ \\
\midrule
\multirow{18}[8]{*}{$T=2$} & \multirow{9}[4]{*}{WN} & \multirow{4}[2]{*}{Non-stud.} & AA    & 5.48  & 5.54  & 5.08  & 4.76  & 5.16 \\
      &       &       & BS    & 5.60  & 5.36  & 5.14  & 4.78  & 5.06 \\
      &       &       & PT    & 4.54  & 4.94  & 4.82  & 4.62  & 5.08 \\
      &       &       & PT2   & 1.02  & 0.80  & 0.76  & 0.96  & 1.00 \\
\cmidrule{3-9}      &       & \multirow{5}[2]{*}{Studentized} & AA    & 14.34 & 8.54  & 6.94  & 6.08  & 4.96 \\
      &       &       & BS    & 14.24 & 8.56  & 6.98  & 6.00  & 4.86 \\
      &       &       & BS2   & 2.56  & 4.66  & 5.02  & 5.30  & 4.54 \\
      &       &       & PT    & 5.06  & 5.08  & 5.20  & 5.20  & 4.44 \\
      &       &       & PT2   & 4.50  & 5.02  & 5.16  & 5.28  & 4.44 \\
\cmidrule{2-9}      & \multirow{9}[4]{*}{AR(1)} & \multirow{4}[2]{*}{Non-stud.} & AA    & 5.12  & 6.02  & 5.12  & 5.24  & 5.22 \\
      &       &       & BS    & 5.36  & 5.90  & 5.12  & 5.22  & 5.10 \\
      &       &       & PT    & 4.16  & 5.32  & 4.88  & 5.00  & 5.06 \\
      &       &       & PT2   & 3.46  & 4.58  & 4.32  & 4.62  & 4.54 \\
\cmidrule{3-9}      &       & \multirow{5}[2]{*}{Studentized} & AA    & 14.00 & 8.78  & 6.88  & 6.30  & 5.58 \\
      &       &       & BS    & 14.18 & 8.96  & 6.98  & 6.40  & 5.52 \\
      &       &       & BS2   & 3.58  & 4.86  & 5.16  & 5.38  & 5.08 \\
      &       &       & PT    & 5.12  & 5.36  & 5.16  & 5.34  & 5.00 \\
      &       &       & PT2   & 5.10  & 5.32  & 5.40  & 5.42  & 5.16 \\
\midrule
\multirow{18}[8]{*}{$T=3$} & \multirow{9}[4]{*}{WN} & \multirow{4}[2]{*}{Non-stud.} & AA    & 5.68  & 4.80  & 5.46  & 5.16  & 5.56 \\
      &       &       & BS    & 5.84  & 4.74  & 5.38  & 5.22  & 5.52 \\
      &       &       & PT    & 5.30  & 4.66  & 5.38  & 5.00  & 5.38 \\
      &       &       & PT2   & 0.88  & 0.72  & 0.60  & 0.94  & 1.10 \\
\cmidrule{3-9}      &       & \multirow{5}[2]{*}{Studentized} & AA    & 35.82 & 16.74 & 9.88  & 7.04  & 5.80 \\
      &       &       & BS    & 36.00 & 16.78 & 9.90  & 6.98  & 5.84 \\
      &       &       & BS2   & 0.70  & 3.28  & 4.66  & 4.82  & 4.92 \\
      &       &       & PT    & 5.02  & 5.08  & 4.98  & 4.84  & 4.88 \\
      &       &       & PT2   & 3.78  & 4.52  & 4.82  & 4.66  & 4.86 \\
\cmidrule{2-9}      & \multirow{9}[4]{*}{AR(1)} & \multirow{4}[2]{*}{Non-stud.} & AA    & 6.22  & 6.18  & 5.76  & 5.20  & 5.64 \\
      &       &       & BS    & 6.08  & 6.24  & 5.74  & 5.42  & 5.38 \\
      &       &       & PT    & 12.70 & 13.80 & 13.88 & 13.68 & 12.96 \\
      &       &       & PT2   & 6.14  & 7.30  & 6.94  & 6.44  & 6.84 \\
\cmidrule{3-9}      &       & \multirow{5}[2]{*}{Studentized} & AA    & 33.44 & 17.06 & 10.38 & 7.26  & 5.96 \\
      &       &       & BS    & 33.94 & 16.92 & 10.12 & 7.22  & 5.80 \\
      &       &       & BS2   & 1.78  & 3.14  & 4.26  & 4.80  & 4.98 \\
      &       &       & PT    & 4.62  & 5.26  & 5.08  & 5.14  & 5.12 \\
      &       &       & PT2   & 3.36  & 4.54  & 4.84  & 4.86  & 5.02 \\
\hline \hline
\end{tabular}
}
\caption{Empirical rejection rates (in \%) under $H_0$ in \eqref{eq:H0} for $\alpha = 5\%$ based on $S=5,000$ i.i.d.\ panel datasets for the various designs.}
\label{tab:panel_size}
\end{table}

Table \ref{tab:panel_power} explores the performance of the same test for data configurations that do not satisfy marginal homogeneity. To make the comparison fair, we focus on asymptotically valid inference methods. For $T=2$, this includes studentized and non-studentized versions of the AA, BS, and PT tests, and studentized BS2. Our results indicate that studentized tests are considerably more powerful than their corresponding non-studentized versions. The main difference in the case of $T=3$, is that we now have to eliminate the non-studentized PT test. With this exception, the results with $T=3$ are qualitatively similar to those for $T=2$.

\begin{table}[htbp]
\centering
\scalebox{1}{
\begin{tabular}{ccccrrrrr}
\hline \hline
$T$     &  $\Xi$ & Test type      & Critical value & $n=30$  & $n=60$  & $n=120$ & $n=240$ & $n=480$ \\
\midrule
\multirow{18}[7]{*}{$T=2$} & \multirow{9}[4]{*}{WN} & \multirow{4}[2]{*}{Non-stud.} & AA    & 8.04  & 11.06 & 16.22 & 32.12 & 65.34 \\
      &       &       & BS    & 8.06  & 11.00 & 16.34 & 32.12 & 65.08 \\
      &       &       & PT    & 6.66  & 10.20 & 15.78 & 31.68 & 64.98 \\
      &       &       & PT2   & 1.30  & 2.14  & 3.44  & 7.64  & 28.86 \\
\cmidrule{3-9}      &       & \multirow{5}[2]{*}{Studentized} & AA    & 19.88 & 18.74 & 27.12 & 46.40 & 78.66 \\
      &       &       & BS    & 20.04 & 18.70 & 27.14 & 46.50 & 78.86 \\
      &       &       & BS2   & 4.92  & 10.96 & 22.38 & 43.90 & 77.68 \\
      &       &       & PT    & 8.32  & 11.96 & 22.60 & 43.94 & 77.72 \\
      &       &       & PT2   & 7.66  & 11.74 & 22.42 & 43.82 & 77.84 \\
\cmidrule{2-9}      & \multirow{9}[3]{*}{AR(1)} & \multirow{4}[2]{*}{Non-stud.} & AA    & 6.50  & 7.64  & 9.54  & 18.14 & 39.56 \\
      &       &       & BS    & 6.58  & 7.76  & 9.22  & 18.14 & 39.70 \\
      &       &       & PT    & 5.36  & 7.08  & 8.90  & 17.80 & 39.40 \\
      &       &       & PT2   & 4.48  & 6.10  & 7.90  & 15.68 & 36.88 \\
\cmidrule{3-9}      &       & \multirow{5}[1]{*}{Studentized} & AA    & 18.82 & 17.60 & 22.52 & 40.96 & 72.34 \\
      &       &       & BS    & 18.96 & 17.76 & 22.50 & 40.90 & 72.24 \\
      &       &       & BS2   & 5.16  & 10.96 & 18.26 & 38.30 & 71.30 \\
      &       &       & PT    & 7.48  & 11.48 & 18.48 & 38.36 & 71.06 \\
      &       &       & PT2   & 7.46  & 11.52 & 18.52 & 38.44 & 71.22 \\
      \midrule
\multirow{18}[7]{*}{$T=3$} & \multirow{9}[3]{*}{WN} & \multirow{4}[1]{*}{Non-stud.} & AA    & 6.98  & 8.50  & 12.70 & 27.70 & 72.42 \\
      &       &       & BS    & 7.16  & 8.54  & 12.64 & 27.78 & 72.56 \\
      &       &       & PT    & 6.32  & 8.14  & 12.18 & 27.26 & 72.44 \\
      &       &       & PT2   & 1.00  & 1.34  & 1.58  & 3.78  & 18.86 \\
\cmidrule{3-9}      &       & \multirow{5}[2]{*}{Studentized} & AA    & 51.14 & 46.10 & 62.68 & 89.28 & 99.66 \\
      &       &       & BS    & 51.22 & 45.94 & 62.56 & 89.20 & 99.68 \\
      &       &       & BS2   & 1.86  & 16.14 & 47.42 & 85.40 & 99.60 \\
      &       &       & PT    & 11.30 & 21.28 & 48.82 & 85.74 & 99.60 \\
      &       &       & PT2   & 9.20  & 19.60 & 47.88 & 85.52 & 99.54 \\
\cmidrule{2-9}      & \multirow{9}[4]{*}{AR(1)} & \multirow{4}[2]{*}{Non-stud.} & AA    & 6.38  & 7.16  & 7.56  & 10.44 & 21.30 \\
      &       &       & BS    & 6.34  & 7.10  & 7.64  & 10.58 & 21.28 \\
      &       &       & PT    & 14.26 & 16.64 & 19.74 & 28.10 & 65.38 \\
      &       &       & PT2   & 6.42  & 8.42  & 9.38  & 12.98 & 28.18 \\
\cmidrule{3-9}      &       & \multirow{5}[2]{*}{Studentized} & AA    & 62.84 & 72.54 & 94.60 & 99.94 & 100.00 \\
      &       &       & BS    & 63.34 & 72.50 & 94.60 & 99.94 & 100.00 \\
      &       &       & BS2   & 6.24  & 33.40 & 86.92 & 99.88 & 100.00 \\
      &       &       & PT    & 16.30 & 45.90 & 89.34 & 99.88 & 100.00 \\
      &       &       & PT2   & 13.00 & 43.08 & 88.52 & 99.88 & 100.00 \\
\hline \hline
\end{tabular}
}
\caption{Empirical rejection rates (in \%) under $H_1$ in \eqref{eq:H0} based on $\alpha_t = \sqrt{1 + (t-1)/2}$ for $t=1,\dots,T$, with $\alpha = 5\%$ based on $S=5,000$ i.i.d.\ panel datasets for the various designs.}
\label{tab:panel_power}
\end{table}

Tables \ref{tab:panel_size} and \ref{tab:panel_power} offer interesting conclusions regarding the finite sample properties of the asymptotically valid tests. When the sample size is relatively small, the non-studentized tests appear to provide better size control than their studentized counterparts, though at the expense of lower power. As the sample size increases, the studentized tests improve their size control. As a result, the studentized tests seem to be a better option for larger sample sizes: they offer adequate size control and higher power than their non-studentized counterparts.

\section{Empirical application} \label{sec:Application}

This section applies our marginal homogeneity tests to the state variable in the dynamic discrete choice game in \cite{igami/yang:2016}. In this paper, the authors develop and estimate a dynamic entry model oligopoly game among Canada's five main hamburger chains: A\&W, Burger King, Harvey's, McDonald's, and Wendy's. They use yearly data from 400 geographical markets\footnote{The paper defines a market as a cluster of stores located within a 0.5-mile radius at any point of their sample period. Markets in downtown areas are omitted because they experience a different form of competition.} located in seven major Canadian cities between 1970 and 2004, i.e., $n=400$ and $T=35$. For each market-year pair $(i,t)$, they observe the number of stores for each chain, population, and income.

The state variable used in \cite{igami/yang:2016} is a discrete categorical variable whose value represents the number of stores of each chain, population, and income of a given market-year pair. We now briefly explain its construction, and defer to their \citet[Section 4.1]{igami/yang:2016} for details. The paper restricts the number of stores per chain to three, and divides population and income into quartiles. 

For each period $t=1,\dots,T=35$ and market $i=1,\dots,n=400$, $X_{i,t}$ is uniquely determined by the number of stores (up to three) for each chain $j=1,2,3,4,5$, ${N}_{i,t,j} \in \{0,1,2,3\}$, the population quartile $P_{i,t} \in \{1,2,3,4\}$, and the income quartile $I_{i,t}\in \{1,2,3,4\}$. So, $X_{i,t}=1$ indicates that ${N}_{i,t,j}=0$ for all $j=1,2,3,4,5$, $P_{i,t}=1$, and $I_{i,t}=1$; $X_{i,t}=2$ indicates that ${N}_{i,t,j}=0$ for all $j=1,2,3,4,5$, $P_{i,t}=1$, and $I_{i,t}=2$, and so on, until $X_{i,t}=4^7=16,384$ indicates that ${N}_{i,t,j}=3$ for all $j=1,2,3,4,5$, $P_{i,t}=4$, and $I_{i,t}=4$. While $X_{i,t}$ could take up to $16,384$ possible values, it only takes $467$ distinct values in the entire dataset. Our goal is to assess whether the marginal distribution of the discrete state variable $X_{i,t}$ remains stable for a chosen period of time.\footnote{Because income and population are discretized, testing with these discretized states naturally raises a post-selection inference problem. This is a common issue because structural model estimation typically relies on discretized states. We omit a detailed treatment here, as it lies beyond the scope of the paper.}

The data spans an extensive period in which the Canadian fast-food industry grew considerably. As \citet[Section 3.2]{igami/yang:2016} reports, the average number of shops per market ranged from less than 0.5 in the 1970s to approximately 1.8 in the early 2000s. We now provide further evidence about the evolution of this industry. Figure \ref{fig:Figure3} shows the average number of competitors per market over time. This figure reveals that the frequency of empty markets decreased steadily between 1970 and 2000, while the frequency of monopoly, duopoly, and triopoly increased more steadily over the same period. In contrast, during 2000-2004, the frequency of each market type has remained relatively stable. This evidence suggests that the Canadian fast-food industry evolved between 1970 and the early 2000s and may have reached a steady state in the last years of the sample. The hypothesis tests developed in this paper can be used to evaluate whether the state distribution is homogeneous across any period in the sample.

\begin{figure}
    \centering
  \scalebox{0.55}{\includegraphics{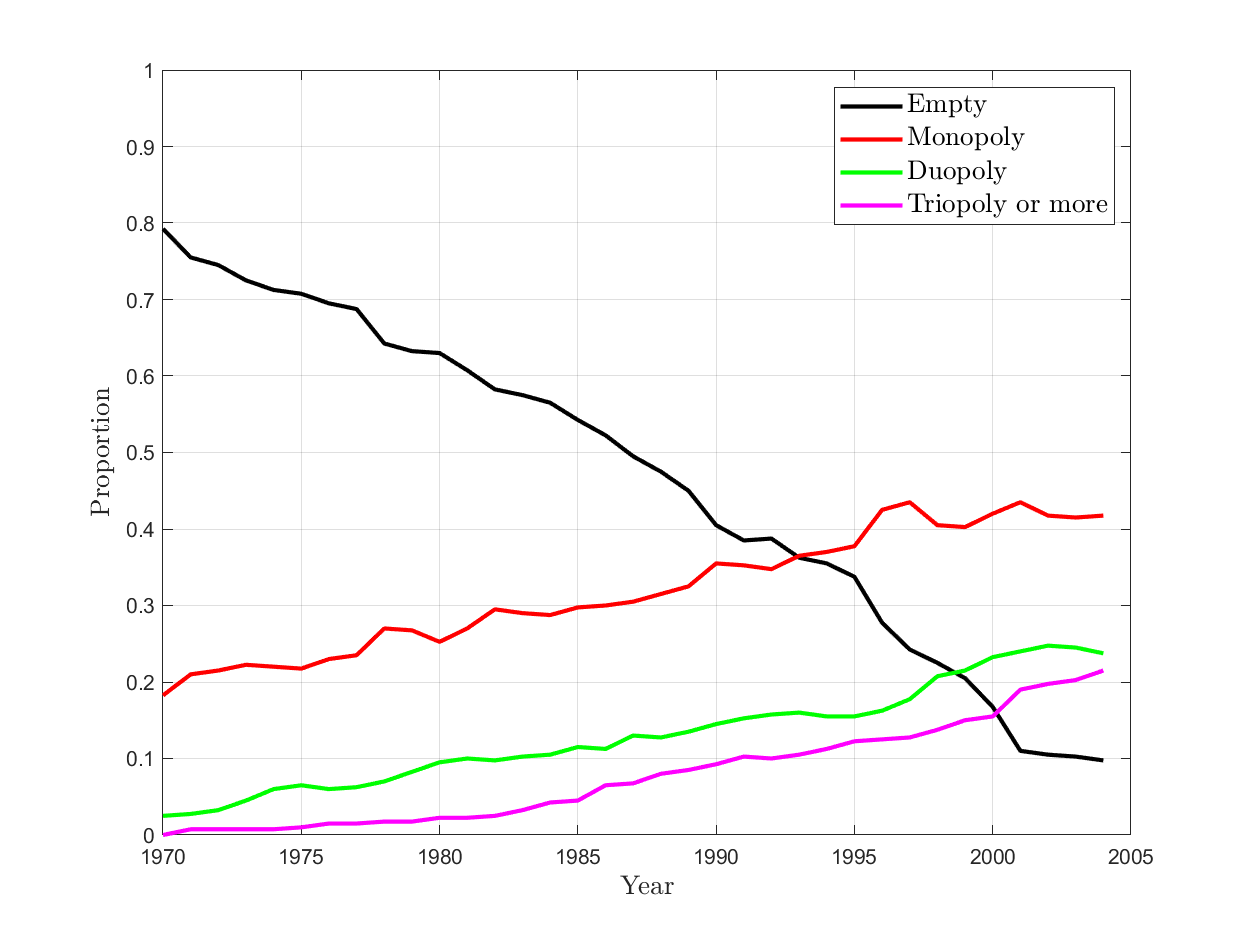}}
    \caption{\small Average type of market over time.}
    \label{fig:Figure3}
\end{figure}

As explained in Section \ref{sec:introduction}, the marginal homogeneity of the state variable can be a source for efficiency gains in the estimation of dynamic discrete games. With this motivation in mind, we now apply our marginal homogeneity tests to our panel data of the state variable. Given the large number of values that the variable takes, Assumption \ref{ass:NS} does not hold over the sample period. For this reason, we only consider non-studentized tests. We implement our tests for two subsets of periods. Since we consider panel data with $T>2$, the non-studentized asymptotic approximation and bootstrap tests are valid, but the permutation test is not. Table \ref{tab:igami} presents the results of our hypothesis tests for two subsets of sample periods. First, we consider a subset of our data every five years, i.e., 1970, 1975, 1980, 1985, 1990, 1995, and 2000. In this case, our tests strongly reject the hypothesis of marginal homogeneity. This result is expected, as it is consistent with the informal discussion in the previous paragraph regarding the growth of the Canadian fast-food industry between 1970 and 2000. Second, we repeat the analysis for the last four years in the sample period, i.e., 2001, 2002, 2003, and 2004. In this case, our tests do not reject the hypothesis of marginal homogeneity. These results suggest that the Canadian fast-food industry may be in a steady state in the latter part of the sample period.

\begin{table}[ht]
\centering
\begin{tabular}{cccccc}
  \hline   \hline
Sample & Test type   & Test statistic &  \multicolumn{3}{c}{Critical Value with $\alpha = 5\%$} \\ 
 & && Asy. approx. & Bootstrap & Permutation \\ 
  \hline
{Pre-2000} & Non-stud.  & 10.76 & 0.88     &    0.82   &       2.35 \\
{Post-2000} & Non-stud.  & 0.08 & 0.13     &    0.13   &       0.23 \\
   \hline   \hline
\end{tabular}
\caption{Marginal homogeneity tests applied to panel data of the state variable in \cite{igami/yang:2016} with $\alpha = 5\%$. The Pre-2000 sample refers to the subsample with $t \in \{1970, 75,80,90,95,00\}$. The Post-2000 sample refers to the subsample with $t \in \{2001,02,03,04\}$. 
}
\label{tab:igami}
\end{table}

\section{Conclusions} \label{sec:Conclusions}

This paper proposes hypothesis tests to evaluate whether panel data satisfy marginal homogeneity. As we argue in the paper, marginal homogeneity is a relevant property in economic settings such as dynamic discrete games, difference-in-differences models, and finance.

Our asymptotic framework for panel data considers a diverging number of units $n$ and a fixed number of periods $T$. We implement our tests by comparing a studentized or non-studentized $T$-sample version of the Cram\'er-von Mises statistic with a suitable critical value. Relative to the non-studentized case, the asymptotic analysis of the studentized statistics requires an additional assumption: the variance–covariance matrix used for studentization must be non-singular. It is worth noting that this condition can fail in practice. In fact, it failed in our empirical application.

We investigate three methods to construct the critical value: asymptotic approximations, the bootstrap, and time permutations. We prove that the asymptotic approximation and bootstrap tests are asymptotically valid, regardless of whether we use studentized or non-studentized test statistics. The permutation test based on a non-studentized statistic is asymptotically exact when $T=2$, but is asymptotically invalid when $T>2$. In contrast, the permutation test based on a studentized statistic is always asymptotically exact. Finally, under a time-exchangeability assumption, the permutation test is valid in finite samples, both with and without studentization.

We also study the power of the various methods. The asymptotically valid tests we consider are consistent and have non-trivial asymptotic power under suitable local alternatives. Moreover, the asymptotically valid tests based on the non-studentized statistics share the same local power properties, and the same holds for the asymptotically valid tests based on the studentized statistics.

Our Monte Carlo simulations investigate the finite sample behavior of our tests. The non-studentized tests exhibit better finite-sample size control than their studentized counterparts, though this comes at the cost of lower power. Finally, we apply our test to the state variable of the dynamic oligopoly model of the Canadian fast-food industry in \cite{igami/yang:2016}. Our findings suggest that the industry evolved between 1970 and 2000, and appears to have reached a steady state since then.

\clearpage
 \appendix
 
 \section{Proofs}\label{sec:appendix_proofs}
 \begin{small}
 \input{appendix_proofs.tex}
 \end{small}

\bibliography{BIBLIOGRAPHY}
\end{document}

%% file: appendix_proofs.tex
Throughout this appendix, we use RHS, LLN, CMT, CLT, PSD, PD, and to abbreviate ``right-hand side'', ``Kolmogorov's strong law of large numbers'', ``continuous mapping theorem'', ``Lindeberg-Levy central limit theorem'', ``positive semi-definite'', and ``positive definite'', respectively. We also define the sequence of random vectors $\{Z_i\}_{i=1}^n$ with 
\begin{equation}
    \label{eq:Zi_defn}
    Z_i ~\equiv~ \left\{\sqrt{{P}(u_k)} [1(X_{i,t} \leq u_k) - 1(X_{i,t+1} \leq u_k)]: \quad t=1,\ldots,T-1, ~~ k=1,\ldots,K\right\} \in \mathbb{R}^{(T-1)K},
\end{equation}
where $P(u_k)$ is the aggregated probability given in \eqref{eq:pop_mass}.

Let $\Omega \in  \mathbb{R}^{(T-1)K \times (T-1)K}$ denote a PSD matrix with a vector of eigenvalues $\ell \in  \mathbb{R}^{(T-1)K}$. By the Principal-Axis Theorem \cite[pp. 397, 418]{scheffe:1959}, $\epsilon'\Omega\epsilon$ with $\epsilon \sim N(0_{(T-1)K\times 1}, \Omega)$, has the same distribution as $S(\ell)$ defined in Section \ref{sec:AA}. That is, $\epsilon'\Omega\epsilon$ has a generalized chi-square distribution of weights equal to $\ell$, unit vector of degrees of freedom, zero vector of non-centrality parameters, and no constant or normal terms. Throughout this appendix, we use $c^\chi(1-\alpha; \Omega)$ to denote the $(1-\alpha)$-quantile of $\epsilon'\Omega\epsilon$.

\begin{proof}[Proof of Theorem \ref{thm:asyDistSample}]
\underline{Part (a).} Under Assumption \ref{ass:A1} and $H_0$, $\{Z_i\}_{i=1}^{n}$ are i.i.d.\ with $E[Z_i]=0_{1 \times (T-1)K}$ and $V[Z_i]=\Sigma_Z$ as defined in \eqref{eq:defnSigma}. By the CLT,
\begin{equation}
\tilde{Z} ~\equiv~ \frac{1}{\sqrt{n}} \sum_{i=1}^{n}{Z}_{i}' ~
\overset{d}{\to }~\xi 
~\sim~ N( 0_{(T-1)K \times 1},\Sigma_Z). \label{eq:LLCLT}
\end{equation}

For any $t=1,\dots,T-1$ and $k=1,\dots,K$, the $((t-1)K + k)$-component of $\hat{\delta}\equiv 
\hat{Z} - \tilde{Z}$ satisfies 
\begin{align}
    &\sqrt{n\hat{P}(u_k)}[\hat{F}_t(u_k) -\hat{F}_{t+1}(u_k) ] - \sqrt{n{P}(u_k)}[\hat{F}_t(u_k) -\hat{F}_{t+1}(u_k) ]\notag\\
    &\overset{(1)}{=}~\sqrt{n}[(\hat{F}_t(u_k) - F_t(u_k)) - (\hat{F}_{t+1}(u_k) - F_{t+1}(u_k))]~ (\sqrt{\hat{P}(u_k) }- \sqrt{P(u_k) }) \notag\\
    &\overset{(2)}{=} ~O_{p}(1) o_p(1)~ =~ o_{p}(1),\label{eq:resid}
\end{align}
where (1) holds by $F_t(u_k)=F_{t+1}(u_k)$, which is implied by $H_0$, and (2) by $\sqrt{n}[(\hat{F}_t(u_k) - F_t(u_k)) - (\hat{F}_{t+1}(u_k) - F_{t+1}(u_k))] = O_{p}(1)$ and $\hat{P}(u_k)-P(u_k) = o_{p}(1)$, which are implied by the CLT and LLN, respectively.

Then, consider the following derivation.
\begin{align}
{S}_{n}
~\overset{(1)}{=}~ 
\tilde{Z}'\tilde{Z} + 2\hat{\delta}'\tilde{Z} + \hat{\delta}'\hat{\delta}
~\overset{(2)}{\overset{d}{\to }}~\xi ^{\prime }\xi.\label{eq:expansion}
\end{align}
where (1) holds by $\hat{\delta} = \hat{Z} - \tilde{Z}$, and (2) by \eqref{eq:LLCLT} and \eqref{eq:resid}.

To complete the proof, it then suffices to show that the RHS of \eqref{eq:expansion} can be expressed as \eqref{eq:LimitingS}. To this end, consider an orthogonal decomposition of $\Sigma_Z = H^{\prime }\Lambda H$, where $H\in \mathbb{R}^{(T-1)K\times (T-1)K}$ is an orthogonal matrix (i.e., $HH^{\prime }=I_{(T-1)K\times (T-1)K}$) and $\Lambda =diag\{\{\lambda_{j}\}_{j=1}^{(T-1)K}\}$ is the diagonal matrix of eigenvalues of $\Sigma_Z$. The desired result then holds by the following derivation:
\begin{equation}
\xi ^{\prime }\xi ~\overset{(1)}{=}~\epsilon ^{\prime }\epsilon ~\overset{(2)}{=}~\sum_{j=1}^{(T-1)K}\lambda _{j}\zeta _{j}^{2},
\label{eq:chi2argument}
\end{equation}
where (1) holds for $\epsilon \equiv H\xi $ which implies $\epsilon ^{\prime}\epsilon =\xi ^{\prime }H^{\prime }H\xi =\xi ^{\prime }\xi $, and (2) by $\epsilon \sim N(0_{(T-1)K \times 1},\Lambda )$, and so $\epsilon =\{ \lambda _{j}^{1/2}\zeta _{j}\} _{j=1}^{(T-1)K}$ for $\zeta \sim N(0_{(T-1)K}, I_{(T-1)K \times (T-1)K})$. 

\noindent \underline{Part (b).} Under Assumption \ref{ass:A1} and the LLN, $\hat{\Sigma}_Z = \Sigma_Z +o_{p}(1)$. Then, Assumption \ref{ass:NS} and the CMT imply that $\hat{\Sigma}_Z^- = \Sigma_Z^{-1} +o_{p}(1)$. We can then repeat the arguments in part (a) to get
\begin{align}
    \hat{\delta} ~=~ o_p(1)~~~\text{ and }~~~\tilde{Z}' \hat{\Sigma}_Z^- \tilde{Z} ~\overset{d}{\to}~  \xi'\Sigma_Z^{-1}\xi  \sim \chi^2_{(T-1)K}.\label{eq:resid2}
\end{align}
From here, the desired result follows from the next derivation:
\begin{align*}
    \bar{S}_n
    ~\overset{(1)}{=}~ 
    \tilde{Z}' \hat{\Sigma}_Z^- \tilde{Z} +2 \hat{\delta}'\hat{\Sigma}_Z^- \tilde{Z}     + \hat{\delta}'\hat{\Sigma}_Z^-\hat{\delta}~\xrightarrow{d}~ \xi'\Sigma_Z^{-1}\xi ~\overset{(3)}{\sim}~ \chi^2_{(T-1)K},
\end{align*}
where (1) holds by $\hat{\delta}= \hat{Z} - \tilde{Z}$, (2) by \eqref{eq:resid2}, and (3) by $\xi \sim N( 0_{(T-1)K},\Sigma_Z)$.
\end{proof}


\begin{proof}[Proof of Theorem \ref{thm:AA_fixed}] 

\noindent \underline{Part (a).}
We divide the argument into two cases. 

\underline{Case 1: $\Sigma_Z \neq 0_{(T-1)K\times (T-1)K}$}. 
Let $\lambda \in \mathbb{R}^{(T-1)K}$ denote the vector of eigenvalues of $\Sigma_Z$. Since $\lambda $ is a continuous function of $\Sigma_Z$ and $\hat{\Sigma}_Z \overset{p}{\to }\Sigma_Z$ holds, the CMT implies that $\hat{\lambda}_{n}\overset{p}{\to }\lambda $. Since $\Sigma_Z$ is positive semi-definite, $\Sigma_Z \neq 0_{(T-1)K\times (T-1)K}$ implies that  $\lambda \neq 0_{(T-1)K\times 1}$. 

Note that $G(x,\ell)$ is continuous in $\ell$ for any $x\in \mathbb{R}$ and $\ell\in \mathbb{R}^{(T-1)K}\backslash 0_{(T-1)K\times 1}$. To see why, consider an arbitrary $x\in \mathbb{R}$ and sequence $\ell^{(n)}\rightarrow \ell\in \mathbb{R}^{(T-1)K}\backslash 0_{(T-1)K\times 1}$. The characteristic function of $S(\ell)$ is $$\varphi ( t,\ell) ~=~\prod_{j=1}^{(T-1)K }( 1-it\ell_{j}) ^{-1/2}$$ and satisfies $\varphi ( t,\ell^{(n)}) \to \varphi ( t,\ell) $ for all $t\in \mathbb{R} $. From this and Levy's Continuity Theorem (e.g., see \citet[Theorem 22.17]{davidson:1994}) we deduce that $S( \ell^{(n)}) \overset{d}{ \to }S( \ell) $. This is equivalent to $G(x,\ell^{(n)})\to G(x,\ell)$ since $S(\ell) $ is continuously distributed (as $\ell\neq 0_{(T-1)K\times 1}$). Since the choices of $x\in \mathbb{R} $ and $\ell \in \mathbb{R} ^{(T-1)K}\backslash 0_{(T-1)K\times 1} $ were arbitrary, the desired result follows. 

Given the continuity of $G(x,\cdot)$ for all $ x\in \mathbb{R} $, the CMT gives $G(x,\hat{\lambda}_{n}) \overset{p}{\to }G(x,\lambda )=P( S\leq x) $. In turn, since $G(x,\hat{\lambda}_{n})$ and $P( S\leq x)$ are weakly increasing and bounded, and $P( S\leq x)$ is continuous (as $\lambda \neq 0_{(T-1)K\times 1}$), an argument along the lines of \citet[Lemma 2.11]{vandervaart:1998} implies that
\begin{equation}
\sup_{x\in \mathbb{R} }\vert G(x,\hat{\lambda}_{n})-P( S\leq x) \vert ~\overset{p}{\to }~0. \label{eq:AA1}
\end{equation}

We now show that
\begin{equation}
    {c}_{n}^{A}(1-\alpha ) ~\overset{p}{\to}~ c^\chi(1-\alpha; \Sigma_Z).\label{eq:cn_consistency}
\end{equation}
Fix $\varepsilon >0$ arbitrarily. Since the CDF of $S$ is continuous and strictly increasing at $c^\chi(1-\alpha; \Sigma_Z)>0$, $\exists\delta =\delta ( \varepsilon ) >0$ such that
\begin{align}
c^\chi(1-\alpha; \Sigma_Z)-c^\chi(1-\alpha-\delta; \Sigma_Z) ~\leq~ \varepsilon ~~~\text{ and }~~~
c^\chi(1-\alpha+\delta; \Sigma_Z) +\delta -c^\chi(1-\alpha; \Sigma_Z) ~\leq~ \varepsilon .\label{eq:quantile_6}
\end{align}
Then, let $E_{n} $ be defined by
\begin{equation*}
E_{n}~=~\left\{ \sup_{x\in \mathbb{R}}|G(x,\hat{\lambda}_{n})-P( S\leq x)|<\delta \right\} .
\end{equation*}
Under $E_{n}$, we have
\begin{equation*}
\delta ~\overset{(1)}{\geq }~
G({c}_{n}^{A}(1-\alpha ),\hat{\lambda}_{n})-P(S \leq {c}_{n}^{A}(1-\alpha ))~\overset{(2)}{\geq}~( 1-\alpha ) -P(S \leq {c}_{n}^{A}(1-\alpha )),
\end{equation*}
where (1) holds by $E_{n}$, and (2) by \eqref{eq:cnAA}, as it implies $G({c}_{n}^{A}(1-\alpha ),\hat{\lambda}_{n})\geq 1-\alpha $. This yields
\begin{equation}
P(S \leq {c}_{n}^{A}(1-\alpha ))~\geq~ ( 1-\alpha ) -\delta . \label{eq:quantile_7}
\end{equation}
From here, we can get
\begin{equation}
{c}_{n}^{A}(1-\alpha )~\overset{(1)}{\geq }~c^\chi(1-\alpha-\delta; \Sigma_Z) ~\overset{(2)}{\geq }~c^\chi(1-\alpha; \Sigma_Z)-\varepsilon , \label{eq:quantile_9}
\end{equation}
where (1) holds by \eqref{eq:quantile_7} and (2) by the first condition in \eqref{eq:quantile_6}. Also under $E_{n}$, we have
\begin{equation*}
-\delta ~\overset{(1)}{\leq }~
G({c}_{n}^{A}(1-\alpha )-\delta,\hat{\lambda}_{n})-P(S \leq {c}_{n}^{A}(1-\alpha )-\delta)~\overset{(2)}{<}~( 1-\alpha ) -P(S\leq {c}_{n}^{A}(1-\alpha )-\delta ),
\end{equation*}
where (1) holds by $E_{n}$, and (2) by $G({c}_{n}^{A}(1-\alpha )-\delta,\hat{\lambda}_{n})<1-\alpha $. This implies that
\begin{equation}
P(S\leq {c}_{n}^{A}(1-\alpha )-\delta )~<~( 1-\alpha ) +\delta . \label{eq:quantile_8}
\end{equation}
From here, we can get
\begin{equation}
{c}_{n}^{A}(1-\alpha )~\overset{(1)}{\leq }~c^\chi (( 1-\alpha ) +\delta ; \Sigma_Z) +\delta ~\overset{(2)}{\leq }~c^\chi(1-\alpha ; \Sigma_Z)+\varepsilon , \label{eq:quantile_10}
\end{equation}
where (1) holds by \eqref{eq:quantile_8} and (2) by the second condition in \eqref{eq:quantile_6}. By combining \eqref{eq:quantile_9} and \eqref{eq:quantile_10}, we conclude that $ \vert {c}_{n}^{A}(1-\alpha )-c^\chi(1-\alpha; \Sigma_Z)\vert \leq \varepsilon $. From this argument, we deduce that
\begin{equation}
P( E_{n}) ~\leq~ P( \vert {c}_{n}^{A}(1-\alpha )-c^\chi(1-\alpha; \Sigma_Z)\vert \leq \varepsilon ) .\label{eq:quantile_11}
\end{equation}
Since $P( E_{n}) \to 1$ by \eqref{eq:AA1}, we conclude from \eqref{eq:quantile_11} that $P( \vert {c}_{n}^{A}(1-\alpha )-c^\chi(1-\alpha; \Sigma_Z)\vert \leq \varepsilon ) \to 1$. Since the choice of $\varepsilon >0$ was arbitrary, \eqref{eq:cn_consistency} follows.

For any $\varepsilon \in ( 0,c^\chi(1-\alpha; \Sigma_Z) ) $, consider the following argument.
\begin{align}
& P( \{ S_{n}<c^\chi(1-\alpha; \Sigma_Z) -\varepsilon \} ) +P( \{ \vert {c}_{n}^{A}( 1-\alpha ) -c^\chi(1-\alpha; \Sigma_Z) \vert \leq \varepsilon \} ) -1\notag  \\
& ~\leq ~P( \{ S_{n}<c^\chi(1-\alpha; \Sigma_Z) -\varepsilon \} \cap \{ \vert {c}_{n}^{A}( 1-\alpha ) -c^\chi(1-\alpha; \Sigma_Z) \vert \leq \varepsilon \} )\notag  \\
& ~\leq ~ \left\{ 
\begin{array}{c}
P( \{ S_{n}<c^\chi(1-\alpha; \Sigma_Z) -\varepsilon \} \cap \{ \vert {c}_{n}^{A}( 1-\alpha ) -c^\chi(1-\alpha; \Sigma_Z) \vert \leq \varepsilon \} ) \notag \\
+P( \{ S_{n}<{c}_{n}^{A}( 1-\alpha ) \} \cap \{ \vert {c}_{n}^{A}( 1-\alpha ) -c^\chi(1-\alpha; \Sigma_Z) \vert >\varepsilon \} )
\end{array}
\right\} \notag  \\
& ~\leq ~\left\{ 
\begin{array}{c}
P( \{ S_{n}<{c}_{n}^{A}( 1-\alpha ) \} \cap \{ \vert c_{n}^{A}( 1-\alpha ) -c^\chi(1-\alpha; \Sigma_Z) \vert \leq \varepsilon \} ) + \\
P( \{ S_{n}<c_{n}^{A}( 1-\alpha ) \} \cap \{ \vert c_{n}^{A}( 1-\alpha ) -c^\chi(1-\alpha; \Sigma_Z) \vert >\varepsilon \} )
\end{array}
\right\} \notag  \\
& ~=~P( S_{n}<c_{n}^{A}( 1-\alpha ) ) \notag \\
& ~\overset{(1)}{\leq }~1-E[ \phi _{n}^{A }( \alpha ) ] \notag \\
& ~\overset{(2)}{\leq }~P( S_{n}\leq c_{n}^{A}( 1-\alpha ) ) \notag \\
& ~=~\left\{ 
\begin{array}{c}
P( \{ S_{n}\leq c_{n}^{A}( 1-\alpha ) \} \cap \{ \vert c_{n}^{A}( 1-\alpha ) -c^\chi(1-\alpha; \Sigma_Z) \vert \leq \varepsilon \} ) + \\
P( \{ S_{n}\leq c_{n}^{A}( 1-\alpha ) \} \cap \{ \vert c_{n}^{A}( 1-\alpha ) -c^\chi(1-\alpha; \Sigma_Z) \vert >\varepsilon \} )
\end{array}
\right\} \notag \\
& ~\leq~ P( S_{n}\leq c^\chi(1-\alpha; \Sigma_Z) +\varepsilon ) +P( \vert c_{n}^{A}( 1-\alpha ) -c^\chi(1-\alpha; \Sigma_Z) \vert > \varepsilon ) ,\label{eq:cv_asy_size}
\end{align}
where (1) and (2) hold by $\{ S_{n}<c_{n}^{A}( 1-\alpha ) \} \subseteq \{ \phi _{n}^{A }(\alpha ) =0\} \subseteq \{ S_{n}\leq c_{n}^{A}( 1-\alpha ) \} $. By taking sequential limits on \eqref{eq:cv_asy_size} as $n\to \infty $ and $\varepsilon \to 0$, and combined with \eqref{eq:cn_consistency}, Theorem \ref{thm:asyDistSample}(a), and the fact that $S$ is continuously distributed, we conclude that $ E[ \phi _{n}^{A}( \alpha ) ] \to P( S> c^\chi(1-\alpha; \Sigma_Z) ) = \alpha$, as desired.

\underline{Case 2: $\Sigma_Z = 0_{(T-1)K\times (T-1)K}$}. Under $H_0$ in \eqref{eq:H0}, we also have $E[Z_i] = 0_{(T-1)K\times 1}$, and so  $Z_i = 0_{(T-1)K\times 1}$ a.s. By \eqref{eq:Zi_defn}, this implies that 
$1(X_{i,t} \leq u_k) = 1(X_{i,t+1} \leq u_k)$  a.s.\ for all $t = 1,\ldots,T-1$ and $k = 1,\ldots K$. In turn, this gives that $\hat{Z} = 0_{(T-1)K\times 1}$ a.s., with $\hat{Z}$ defined as in \eqref{eq:Z_hat_main_tex}. Then, $S_n = 0$ holds a.s. By this and $c^A_n(1-\alpha) \geq 0$, we get $E[\phi_n^A(\alpha)] = P(S_n > c^A_n(1-\alpha)) = 0 < \alpha$. From here, the desired result holds by taking limits as $n\to\infty$.

\noindent  \underline{Part (b)} The conclusion follows directly from combining \eqref{eq:testAA_norm}, Theorem \ref{thm:asyDistSample}(b), and that the CDF of $\chi_{(T-1)K}^2$ is continuous at $c^\chi(1-\alpha; I_{(T-1)K \times (T-1)K})>0$.
\end{proof}


\begin{proof}[Proof of Theorem \ref{thm:B_fixed}]
\underline{Part (a)}. We divide the argument into two cases. 

\noindent \underline{Case 1: $\Sigma_Z \neq 0_{(T-1)K\times (T-1)K}$}. 
For each $ i=1,\ldots ,n$, $t = 1,\ldots, T-1$, and $k = 1,\ldots, K$, let
\begin{equation*}
{Z}^*_{i,(t-1)K+k} ~\equiv~ \sqrt{{P}(u_k)} [1(X^*_{i,t} \leq u_k) - \hat{F}_t(u_k) - (1(X^*_{i,t+1} \leq u_k) - \hat{F}_{t+1}(u_k))],
\end{equation*}
and let $Z^*_i = \{{Z}^*_{i,(t-1)K+k}:k=1,\dots,K,t=1,\dots,T-1\} \in \mathbb{R}^{(T-1)K} $ and $\tilde{Z}^* \equiv \frac{1}{\sqrt{n}} \sum_{i=1}^{n} {{Z}^{*}_{i}}'$. By \citet[Theorem 3.6.2]{vandervaart/wellner:1996},
\begin{equation}
\{\tilde{Z}^*|{\bf X}_n\} ~
\overset{d}{\to }~\xi 
~\sim~ N( 0_{(T-1)K \times 1},\Sigma_Z)~~~\text{a.s.} \label{eq:CLT_boot}
\end{equation}

Let $\hat{\delta}^*\equiv 
\hat{Z}^* - \tilde{Z}^*$ with $\hat{Z}^*$ as in \eqref{eq:Z_hat_star}. For any $t=1,\dots,T-1$ and $k=1,\dots,K$, then conditional on ${\bf X}_n$, the $((t-1)K + k)$-component of $\hat{\delta}^*$ satisfies 
\begin{align}
    & \sqrt{n \hat{P}(u_k)}[\hat{F}_{t}^*(u_k) - \hat{F}_t(u_k) - (\hat{F}_{t+1}^*(u_k) - \hat{F}_{t+1}(u_k))]
    -
    \sqrt{n {P}(u_k)}[\hat{F}_{t}^*(u_k) - \hat{F}_t(u_k) - (\hat{F}_{t+1}^*(u_k) - \hat{F}_{t+1}(u_k))] \notag \\
    & = \sqrt{n}[\hat{F}_{t}^*(u_k) - \hat{F}_t(u_k) - (\hat{F}_{t+1}^*(u_k) - \hat{F}_{t+1}(u_k))] (\sqrt{\hat{P}(u_k)} - \sqrt{P(u_k)}) \notag \\
    & \overset{(1)}{=} O_p(1)o(1) = o_p(1),  ~~\text{a.s.},\label{eq:resid_boot}
\end{align}
where (1) holds by \citet[Theorem 3.6.2]{vandervaart/wellner:1996} (which implies that $\{\sqrt{n}[\hat{F}_{t}^*(u_k) - \hat{F}_t(u_k) - (\hat{F}_{t+1}^*(u_k) - \hat{F}_{t+1}(u_k))] |{\bf X}_n\}= O_p(1)$ a.s.) and that $\hat{P}(u_k)-P(u_k) =o_{a.s.}(1)$ (and, thus, $\hat{P}(u_k)-P(u_k) |{\bf X}_n\}= o_{p}(1)$ a.s.) by LLN.

Then, consider the following derivation. Conditional on ${\bf X}_n$, 
\begin{align}
{S}_{n}^*
~\overset{(1)}{=}~ 
(\tilde{Z}^*)'(\tilde{Z}^*) + 2(\hat{\delta}^*)'(\tilde{Z}^*) + (\hat{\delta}^*)'(\hat{\delta}^*)
~\overset{(2)}{\overset{d}{\to }}~S, ~~\text{w.p.a.1},
\label{eq:expansion_boot}
\end{align}
where (1) holds by $\hat{\delta}^* = \hat{Z}^* - \tilde{Z}^*$ and (2) by \eqref{eq:CLT_boot} and \eqref{eq:resid_boot}. As a corollary of \eqref{eq:expansion_boot} and also by the condition that $\Sigma_Z$ is a nonzero matrix, we deduce that $P(S^*_n\leq x |{\bf X}_n) \overset{p}{\to } P( S\leq x) $ for all points $ x\in \mathbb{R} $. From this point onward, the rest of the proof is identical to that of part (a) in Theorem \ref{thm:AA_fixed}.

\noindent \underline{Case 2: $\Sigma_Z = 0_{(T-1)K\times (T-1)K}$}. This result holds by the same argument as in Theorem \ref{thm:AA_fixed}, except that $\phi_n^{A}(\alpha)$ and $c^{A}_n(1-\alpha)$ are replaced by $\phi_n^B(\alpha)$ and $c^B_n(1-\alpha)$, respectively.

\noindent \underline{Part (b).} 
By the proof of part (b) in Theorem \ref{thm:AA_fixed}, $\hat{\Sigma}_Z^- = \Sigma_Z^{-1} +o_{p}(1)$. Then, conditional on ${\bf X}_n$,
\begin{equation}
    \hat{\Sigma}_Z^- ~\to~ \Sigma_Z^{-1}  ~~\text{w.p.a.1}.
    \label{eq:sigma_boot}
\end{equation}
We can then repeat the arguments in part (a) to get that, conditional on ${\bf X}_n$,
\begin{align}
    \hat{\delta}^{*} ~=~ o_p(1)~~~\text{ and }~~~{(\tilde{Z}^*)}' \hat{\Sigma}_Z^- ({\tilde{Z}^*}) ~\overset{d}{\to}~  \xi'\Sigma_Z^{-1}\xi  \sim \chi^2_{(T-1)K}~~\text{w.p.a.1}.\label{eq:resid_boot_2}
\end{align}
From here, the desired result follows from the next derivation. Conditional on ${\bf X}_n$,
\begin{align}
    \bar{S}_n^* 
    ~&\overset{(1)}{=}~ 
    (\tilde{Z}^*)'\hat{\Sigma}_Z^-(\tilde{Z}^*) + 2(\hat{\delta}^*)'\hat{\Sigma}_Z^-(\tilde{Z}^*) + (\hat{\delta}^*)'\hat{\Sigma}_Z^-(\hat{\delta}^*) 
    ~\overset{(2)}{\overset{d}{\to }}~\chi_{(T-1)K}^2, ~\text{w.p.a.1},
    \label{eq:expansion_boot2}
\end{align}
 where (1) holds by $\hat{\delta}^* = \hat{Z}^* - \tilde{Z}^*$, and (2) by \eqref{eq:sigma_boot} and \eqref{eq:resid_boot_2}. As a corollary of \eqref{eq:expansion_boot2}, we have that $P(\bar{S}^*_n\leq x |{\bf X}_n) \overset{p}{\to } P( \chi_{(T-1)K}^2\leq x) $ for all $ x\in \mathbb{R} $. From this point onward, the rest of the proof follows from arguments in part (a) in Theorem \ref{thm:AA_fixed}.
\end{proof}

\begin{theorem} \label{thm:perm_dist}
Let Assumption \ref{ass:A1} hold, and let $\hat{Z}^{\bm{\pi}}$ be as in \eqref{eq:Z_hat_perm}, where $\bm{\pi}$ is a uniformly chosen \textit{random permutation} in $\mathcal{M}^n$. Then, 
\begin{equation}
    P(\hat{Z}^{\bm{\pi}} \leq x | {\bf X}_n) ~\overset{p}{\to}~P(\xi \leq x) 
    \label{eq:chungRomanoLemmaA1}
\end{equation}
for all $x$ such that $P(\xi \leq x) $ is continuous, where $\xi \sim N(0_{(T-1)K \times 1}, \Omega^{\pi}_{Z})$ with 
\begin{equation}
    \Omega^{\pi}_{Z} ~\equiv~ \frac{1}{T!}\sum_{\pi \in \mathcal{M}} B(\pi) \Omega_Z B(\pi)', \label{eq:Omega_pi}
\end{equation}
$\Omega_Z = E[Z_iZ_i']$ with $Z_i$ as in \eqref{eq:Zi_defn}, and $\{B({\pi }) \in  \{-1,0,1\}^{(T-1)K\times (T-1)K}:\pi \in \mathcal{M}\}$ are known matrices defined within in the proof.
\end{theorem}
\begin{proof}
We divide the proof into several steps.

\underline{Step 1.} Introduce suitable notation.

For each $i=1,\ldots ,n$ and $k=1,\ldots ,K$, let
\begin{equation}
V_{i,k}~\equiv ~\left[
\begin{array}{c}
1(X_{i,1}\leq u_{k})-1(X_{i,2}\leq u_{k}) \\ \vdots \\
1(X_{i,T-1}\leq u_{k})-1(X_{i,T}\leq u_{k})
\end{array}
\right]~ ~\text{and}~~ V_{i}~\equiv ~
\left[ 
\begin{array}{c}
V_{i,1} \\ 
\vdots  \\ 
V_{i,K}
\end{array}
\right] \in \{-1,0,1\}^{(T-1)K\times 1}. \label{eq:V_i}
\end{equation}
Also, let
\begin{align}
M~ &\equiv ~diag\{ \sqrt{P(u_{1})},\ldots ,\sqrt{P(u_{K})}\} \otimes I_{T-1}~~\in ~[0,1]^{(T-1)K\times (T-1)K} \notag \\ 
\hat{M}~&\equiv ~diag\{ \sqrt{\hat{P}(u_{1})},\ldots ,\sqrt{\hat{P} (u_{K})}\} \otimes I_{T-1}~~\in~ [0,1]^{(T-1)K\times (T-1)K},\label{eq:M}
\end{align}
where $\otimes $ denotes the Kronecker product.  
Note that $Z_i = M \times V_i$ for all $i=1,\dots,n$. 

Let ${\pi }^{n}=\{{\pi }_{i}:i=1,\dots ,n\}\in \mathcal{M}^{n}$ denote a fixed permutation. For any $i=1,\dots ,n$ and $k=1,\dots ,K$, the $\pi$-permutation analogs of $V_{i,k}$ and $V_{i}$
\begin{align}
V_{i,k}^{\pi } \equiv \left[
\begin{array}{c}
1(X_{i,\pi _{i}{(1)}}\leq u_{k})-1(X_{i,\pi _{i}{(2)}}\leq u_{k}) \\
\vdots  \\ 
1(X_{i,\pi _{i}{(T-1)}}\leq u_{k})-1(X_{i,\pi _{i}{(T)}}\leq u_{k})
\end{array}
\right]~ ~\text{and}~~ V_{i}^{\pi }\equiv \left[
\begin{array}{c}
V_{i,1}^{\pi } \\ 
\vdots  \\ 
V_{i,K}^{\pi }
\end{array}
\right] \in \{-1,0,1\}^{K(T-1)\times 1}.\label{eq:perm_defn}
\end{align}
We note that $\hat{M}$ is invariant to ${\pi }^{n}$. To see why, note that for each $k=1,\ldots ,K$, and $t=1,\dots ,T$,
\begin{equation}
\hat{P}^{{\pi }}(u_{k})~=~\frac{1}{nT}\sum_{i=1}^{n} \sum_{t=1}^{T}1(u_{k-1}<X_{i,{\pi }_{i}(t)}\leq u_{k})~\overset{(1)}{=}~ \frac{1}{nT}\sum_{i=1}^{n}\sum_{t=1}^{T}1(u_{k-1}<X_{i,t}\leq u_{k})~=~\hat{P }(u_{k}), \label{eq:P_pi_constant}
\end{equation}
where (1) holds because the sum over $t=1,\dots ,T$ is invariant across the permutation. 

\underline{Step 2.} For any $i=1,\dots,n$ and $\pi_{i}  \in \mathcal{M}$, we define the matrix $B({\pi}_i)$ described in the statement and establish the following representation:
\begin{equation}
V_{i}^{\pi } ~=~B({\pi }_{i}) \times V_{i}. \label{eq:B_perm}
\end{equation}

This result follows from expressing $V_{i,k}^{\pi }$ as a particular linear combination of $V_{i,k}$. To see why, fix $t=1,\dots, T-1$ arbitrarily. If $\pi _{i}(t)<\pi _{i}(t+1)$, then we have
\begin{align*}
1(X_{i,{\pi }_{i}(t)}\leq u_{k})-1(X_{i,{\pi }_{i}(t+1)}\leq u_{k})& =\left[
\begin{array}{c}
1(X_{i,{\pi }_{i}(t)}\leq u_{k})-1(X_{i,{\pi }(t)+1}\leq u_{k})+\ldots \\
+1(X_{i,{\pi }_{i}(t+1)-1}\leq u_{k})-1(X_{i,{\pi }(t+1)}\leq u_{k})
\end{array}
\right]  \\
& =(0,\ldots ,0,1,\ldots ,1,0,\ldots ,0) \times V_{i,k},
\end{align*}
where the ones are located at time periods corresponding to $\pi _{i}(t),\dots ,\pi _{i}(t+1)-1$. Conversely, if $\pi _{i}(t)>\pi _{i}(t+1)$, then we have
\begin{align*}
1(X_{i,{\pi }_{i}(t)}\leq u_{k})-1(X_{i,{\pi }_{i}(t+1)}\leq u_{k})& ~=~\left[ 
\begin{array}{c}
-(1(X_{i,{\pi }_{i}(t)-1}\leq u_{k})-1(X_{i,{\pi }(t)}\leq u_{k}))+\ldots \\
-(1(X_{i,{\pi }_{i}(t+1)}\leq u_{k})-1(X_{i,{\pi }(t+1)+1}\leq u_{k}))
\end{array}
\right]  \\
& ~=~(0,\ldots ,0,-1,\ldots ,-1,0,\ldots ,0)\times V_{i,k},
\end{align*}
where the negative ones are located at time periods corresponding to $\pi _{i}(t+1),\dots,\pi_{i}(t)-1$. Since $\pi _{i}(t)$ and $\pi _{i}(t+1)$ were arbitrarily chosen, we can define a matrix $\tilde{B}({{\pi } _{i}})\in \{-1,0,1\}^{(T-1)\times (T-1)}$ such that $V_{i,k}^{\pi }=\tilde{B}(\pi _{i}) \times V_{i,k}$. By collecting results for $k=1,\ldots K$, and setting 
\begin{equation}
\label{eq:defn_B_mat}
    B(\pi _{i})~\equiv~ diag\{ \tilde{B}(\pi _{i}),\ldots ,\tilde{B}(\pi _{i})\},
\end{equation}
\eqref{eq:B_perm} follows. Finally, by repeating this operation for all $\pi = \pi_i \in \mathcal{M}$, we define the collection of matrices $\{B({\pi }) \in  \{-1,0,1\}^{(T-1)K\times (T-1)K}:\pi \in \mathcal{M}\}$.

\underline{Step 3.} Establish the Hoeffding's condition for $\frac{1}{\sqrt{n }}\sum_{i=1}^{n}B(\bm{\pi }_{i})Z_{i}$, where $\bm{\pi }^{n} = \{\bm{\pi }_{i}: i=1,\dots,n\}$ denotes a randomly chosen permutation in $\mathcal{M}^{n}$. That is,
\begin{equation}
\left( \frac{1}{\sqrt{n}}\sum_{i=1}^{n}B(\bm{\pi }_{i})Z_{i},\frac{1}{\sqrt{ n}}\sum_{i=1}^{n}B(\bm{\tilde{\pi}}_{i})Z_{i}\right) ~\overset{d}{ \to }~( \xi,\tilde{\xi}) , \label{eq:hoeffding}
\end{equation}
where $\bm{\pi}^{n} = \{\bm{\pi }_{i}: i=1,\dots,n\}$ and $\bm{\tilde{\pi}}^{n} = \{\bm{\tilde{\pi} }_{i}: i=1,\dots,n\}$ denote two mutually independent random permutations chosen uniformly from $\mathcal{M}^{n}$ and independent of the data, and $\xi$ and $\tilde{\xi}$ are i.i.d.\ $N(0_{(T-1)K\times 1},\Omega _{Z}^{\pi })$.

We establish \eqref{eq:hoeffding} using the Cram\'{e}r-Wold device. That is, for arbitrary $\lambda ,\nu \in \mathbb{R}^{(T-1)K\times 1}$, \eqref{eq:hoeffding} follows from showing that
\begin{equation}
\frac{1}{\sqrt{n}}\sum_{i=1}^{n}( \lambda ^{\prime }B(\bm{\pi } _{i})+\nu ^{\prime }B(\bm{\tilde{\pi}}_{i})) Z_{i}~\overset{d}{ \to }~N(0_{(T-1)K\times 1},\lambda ^{\prime }\Omega _{Z}^{\pi }\lambda +\nu ^{\prime }\Omega _{Z}^{\pi }\nu ). \label{eq:hoeffding2}
\end{equation}

We begin the argument by showing that $ \{ ( \lambda ^{\prime }B(\bm{\pi }_{i})+\nu ^{\prime }B( \bm{\tilde{\pi}}_{i})) Z_{i}\} _{i=1}^{n}$ is an i.i.d.\ sequence. To see why, note that $\{ Z_{i}\} _{i=1}^{n}$ is i.i.d.\ by Assumption \ref{ass:A1}. Also, since $\bm{\pi}^{n}= \{\bm{\pi }_{i}: i=1,\dots,n\}$ and $\bm{\tilde{\pi}}^{n} = \{\bm{\tilde{\pi} }_{i}: i=1,\dots,n\}$ are defined as i.i.d.\ sequences, we conclude that $\{ B(\bm{\pi }_{i})\} _{i=1}^{n}$ and $\{ B( \bm{\tilde{\pi}}_{i})\} _{i=1}^{n}$ are also i.i.d. By combining these facts, we get that $ \{ ( \lambda ^{\prime }B(\bm{\pi }_{i})+\nu ^{\prime }B( \bm{\tilde{\pi}}_{i})) Z_{i}\} _{i=1}^{n}$ is an i.i.d.\ sequence.

As a next step, we now show that
\begin{equation}
\sum_{\pi \in \mathcal{M} }B( \pi ) ~=~0_{(T-1)K\times (T-1)K}. \label{eq:sumB}
\end{equation}
Since $B(\pi )\equiv diag\{ \tilde{B}(\pi ),\ldots ,\tilde{B}(\pi )\} $, it suffices to show that $\sum_{\pi \in \mathcal{M} }\tilde{B}( \pi) =0_{( T-1) \times (T-1)}$. For each $t =1,\ldots,T-1$, denote the $t$'th row of $ \tilde{B}(\pi)$, denoted $\tilde{B}_{t}(\pi)$, can be expressed as follows:
\[
\tilde{B}_{t}(\pi) ~=~  (0,\dots,0,1,\dots,1,0,\dots,0) \times (1[\pi(t) < \pi(t+1)] - 1[\pi(t)>\pi(t+1)]),
\]
where the sequence of ones appears in the positions $\min\{\pi(t),\pi(t+1)\}$ through $\max\{\pi(t),\pi(t+1)\}-1$. From here, we get that $ \sum_{\pi\in \mathcal{M} }\tilde{B}_{t}( \pi ) =0_{1 \times (T-1)}$, as the occurrence of $(0,\dots,0,1,\dots,1,0,\dots,0)$ in the sum over $\mathcal{M}$ cancels with the corresponding $(0,\dots,0,-1,\dots,-1,0, \dots,0)$ when $\pi(t)$ and $\pi(t+1)$ are reversed.

Next, we show that $E[ B(\bm{\pi }_{i})] =E[ B( \bm{\tilde{\pi}}_{i})] =0_{(T-1)K\times (T-1)K}$ for all $i=1,\ldots ,n $. To see why, fix $i=1,\ldots ,n$ arbitrarily and note that
\begin{equation}
E[ B(\bm{\pi }_{i})] ~\overset{(1)}{=}~E[ B(\bm{\tilde{\pi}}_{i})]~\overset{(2)}{=}~\frac{1}{T!}\sum_{\pi \in \mathcal{M}}B( \pi )~\overset{(3)}{=}~0_{(T-1)K\times (T-1)K}, \label{eq:Bmean0}
\end{equation}
where (1) holds because $B(\bm{\pi }_{i})$ and $B(\bm{\tilde{\pi}} _{i})$ are equally distributed, (2) because there are $|\mathcal{M} |=(T!)$ possible permutations of $\{1,2,\ldots ,T\}$, all equally likely, and (3) by \eqref{eq:sumB}. Then, for all $i=1,\ldots ,n$,
\begin{equation}
    E[ ( \lambda ^{\prime }B(\bm{\pi }_{i})+\nu ^{\prime }B( \bm{\tilde{\pi}}_{i})) Z_{i}] ~\overset{(1)}{=}~E[ ( \lambda ^{\prime }E[ B(\bm{\pi }_{i})] +\nu ^{\prime }E[ B( \bm{\tilde{\pi}}_{i})] ) Z_{i}] ~\overset{(2)}{=}~
    0,
    \label{eq:mean_zero}
\end{equation}
where (1) holds by $Z_{i}\perp (B(\bm{\pi }_{i}),B(\bm{\tilde{\pi}}_{i}))$ and (2) by \eqref{eq:Bmean0}.

From here, note that for all $i=1,\ldots ,n$,
\begin{align}
V[ ( \lambda ^{\prime }B(\bm{\pi }_{i})+\nu ^{\prime }B( \bm{\tilde{\pi}}_{i})) Z_{i}] &~=~E[ ( \lambda ^{\prime }B(\bm{\pi }_{i})+\nu ^{\prime }B(\bm{\tilde{\pi}}_{i})) Z_{i}Z_{i}^{\prime }( B(\bm{\pi }_{i})^{\prime }\lambda +B( \bm{\tilde{\pi}}_{i})^{\prime }\nu ) ] \notag\\
&~\overset{(1)}{=}~E[ ( \lambda ^{\prime }B(\bm{\pi }_{i})+\nu ^{\prime }B( \bm{\tilde{\pi}}_{i})) \Omega _{Z}( B(\bm{\pi }_{i})^{\prime }\lambda +B(\bm{\tilde{\pi}}_{i})^{\prime }\nu ) ] \notag\\
&~\overset{(2)}{=}~\lambda ^{\prime }E [ B(\bm{{\pi}}_{i})\Omega _{Z}B(\bm{{\pi}}_{i})^{\prime } ]\lambda +\nu ^{\prime }E [ B(\bm{\tilde{\pi}}_{i})\Omega _{Z}B(\bm{\tilde{\pi}}_{i})^{\prime } ]\nu ,\notag\\
&~\overset{(3)}{=}~\lambda ^{\prime }\Omega _{Z}^{\pi }\lambda +\nu ^{\prime }\Omega _{Z}^{\pi }\nu ,\label{eq:var_good}
\end{align}
where (1) holds by $Z_{i}\perp (B(\bm{\pi }_{i}),B(\bm{\tilde{\pi}}_{i}))$ and $E [ Z_{i}Z_{i}'] =\Omega _{Z}$, (2) by \eqref{eq:Bmean0} and that $ B(\bm{\pi }_{i}) \perp B( \bm{\tilde{\pi}}_{i})$, and (3) by \eqref{eq:Omega_pi} and that there are $|\mathcal{M} |=(T!)$ possible permutations of $\{1,2,\ldots ,T\}$, and all are equally likely.

To conclude the step, note that \eqref{eq:hoeffding2} follows from the CLT, as $\{ ( \lambda ^{\prime }B(\bm{\pi }_{i})+\nu ^{\prime }B( \bm{\tilde{\pi}}_{i})) Z_{i}\} _{i=1}^{n}$ was shown to be an i.i.d.\ sequence that satisfies \eqref{eq:mean_zero} and \eqref{eq:var_good}.

\underline{Step 4.} Use the previous steps to conclude the proof.

By \citet[Lemma A.1]{chung/romano:2016}, \eqref{eq:chungRomanoLemmaA1} is equivalent to showing that $\hat{Z}^{\pi }$ satisfies the following Hoeffding condition:
\begin{equation}
( \hat{Z}^{\bm{\pi }},\hat{Z}^{\bm{\tilde{\pi}}}) ~\overset{d}{\to}~(\xi,\tilde{\xi}) ,
\end{equation}
where $\hat{Z}^{\bm{\pi }}$ and $\hat{Z}^{\bm{\tilde{\pi}}}$ are permuted according to $\bm{\pi}^{n} = \{\bm{\pi }_{i}: i=1,\dots,n\}$ and $\bm{\tilde{\pi}}^{n} = \{\bm{\tilde{\pi} }_{i}: i=1,\dots,n\}$, respectively, which are two mutually independent random permutations chosen uniformly from $\mathcal{M}^{n}$ and independent of the data, and $\xi$ and $\tilde{\xi}$ are i.i.d.\ according to $N(0_{(T-1)K\times 1},\Omega _{Z}^{\pi })$.

Before proving the desired result, we establish three preliminary results. First, by repeating the arguments in step 3 but with $M $ replaced by $I \in \mathbb{R}^{K(T-1) \times K(T-1)}$, we have that
\begin{equation}
\frac{1}{\sqrt{n}} \sum_{i=1}^{n}B(\bm{\pi }_{i})V_{i} = O_p(1)~~~\text{and} ~~~\frac{1}{ \sqrt{n}}\sum_{i=1}^{n}B(\bm{\tilde{\pi}}_{i})V_{i} = O_p(1). \label{eq:OpWithIV}
\end{equation}
Second, note that Assumption \ref{ass:A1}, the LLN, and the CMT imply that
\begin{equation}
\hat{\delta} ~\equiv~\hat{M}-M~=~o_{p}(1). \label{eq:Mconvergence}
\end{equation}
Third, note that for any permutation $\pi \in \mathcal{M}$, we have
\begin{align}
M\times B(\pi ) 
&~=~diag\{ \sqrt{P(u_{1})}\tilde{B}(\pi ),\ldots ,\sqrt{P(u_{K})} \tilde{B}(\pi )\} \nonumber \\
&~=~\{ diag\{ \tilde{B}(\pi ),\ldots ,\tilde{B}(\pi )\} \}~\times~ ( diag\{ \sqrt{P(u_{1})},\ldots \sqrt{ P(u_{K})}\} \otimes I_{T-1}) \nonumber \\
&~=~B(\pi )\times M. \label{eq:interchange}
\end{align}

The desired result follows from the next derivation.
\begin{align*}
( \hat{Z}^{\bm{\pi }},\hat{Z}^{\bm{\tilde{\pi}}}) &~\overset{(1)}{=}~\left( \frac{{\hat{M}}}{\sqrt{n}}\sum_{i=1}^{n}V_{i}^{\bm{\pi} },\frac{{\hat{M}}}{\sqrt{n}}\sum_{i=1}^{n}V_{i}^{\bm{\tilde{\pi}}} \right) \\
&~\overset{(2)}{=}~\left( \frac{{\hat{M}}}{\sqrt{n}}\sum_{i=1}^{n}B(\bm{\pi } _{i})V_{i},~\frac{{\hat{M}}}{\sqrt{n}}\sum_{i=1}^{n}B(\bm{\tilde{\pi}} _{i})V_{i}\right) \\
&~\overset{(3)}{=}~\left(
\begin{array}{c}
\frac{1}{\sqrt{n}}\sum_{i=1}^{n}B(\bm{\pi }_{i})Z_i+\hat{\delta} \frac{1}{\sqrt{n}}\sum_{i=1}^{n}B(\bm{\pi }_{i})V_{i}, \\
\frac{{1}}{\sqrt{n}}\sum_{i=1}^{n}B(\bm{\tilde{\pi}}_{i})Z_i+\hat{\delta}\frac{1}{\sqrt{n}}\sum_{i=1}^{n}B(\bm{\tilde{\pi}}_{i})V_{i}
\end{array}
\right) \\ 
&~\overset{(4)}{=}~\left(
\frac{1}{\sqrt{n}}\sum_{i=1}^{n}B(\bm{\pi }_{i})Z_i, \frac{1}{\sqrt{n}}\sum_{i=1}^{n}B(\bm{\tilde{\pi}}_{i})Z_i\right)+o_{p}( 1) \\
&~\overset{{(5)}}{\overset{d}{\to}}( \xi,\tilde{\xi}) ,
\end{align*}
as desired, where (1) holds by \eqref{eq:Z_hat_perm}, \eqref{eq:perm_defn}, and \eqref{eq:P_pi_constant}, (2) holds by \eqref{eq:B_perm}, (3) by \eqref{eq:Mconvergence} and $M B( \bm{\pi }_{i})V_{i} = B( \bm{\pi }_{i})Z_{i}$ for all $i=1,\dots,n$ (implied by $Z_{i}=MV_{i}$ for all $i=1,\dots,n$ and \eqref{eq:interchange}), (4) by \eqref{eq:OpWithIV} and \eqref{eq:Mconvergence}, and (5) by \eqref{eq:hoeffding}.
\end{proof}

\begin{proof}[Proof of Theorem \ref{thm:P_fixed}]
Throughout the proof, we continuously invoke the results and notations from Theorem \ref{thm:perm_dist}. Recall that this result derives the asymptotic distribution of $\hat{Z}^{\bm{\pi}}$ as in \eqref{eq:Z_hat_perm}, where $\bm{\pi}$ is a uniformly chosen random permutation in $\mathcal{M}^n$. Recall from this theorem that $\Omega_Z \equiv E[Z_i Z_i']$, which equals $\Sigma_Z \equiv var(Z_i)$ under $H_0$ in \eqref{eq:H0};  $\Omega_Z^\pi$, defined in \eqref{eq:Omega_pi}, represents the asymptotic variance of randomization distribution; $B(\pi)$, defined below \eqref{eq:defn_B_mat} for $\pi\in\mathcal{M}$, denotes a known matrix taking values in $\{-1,0,1\}$. With these notations in mind, we present the proof below.

\noindent \underline{Part (a).} We divide the argument into two cases.

\noindent  \underline{Case 1: $\Sigma_Z \neq 0_{T(K-1)\times T(K-1)}$.} By $T=2$, there are $T!=2$ permutations of $(1,2)$, hence $\mathcal{M} =\{(1,2),(2,1)\}$. Following the construction in step 2 of Theorem \ref{thm:perm_dist}, $B((1,2)) = I_{K\times K}$ and $B((2,1)) = -I_{K\times K}$. Therefore,
\[
    \Omega_Z^\pi ~\overset{(1)}{=}~ \tfrac{1}{2} B(1,2) \Omega_Z B(1,2)' + \tfrac{1}{2} B((2,1)) \Omega_Z B((2,1))' ~\overset{(2)}{=}~ \Omega_Z~\overset{(3)}{=}~ \Sigma_Z,
\]
where (1) holds by the definition of $\Omega_Z^\pi$, (2) holds by the definition of $B(\pi)$, and (3) follows by $\Sigma_Z \equiv var(Z_i) = E[Z_iZ_i'] \equiv \Omega_Z$ under $H_0$ in \eqref{eq:H0}.

Theorem \ref{thm:perm_dist} then implies that  $P(\hat{Z}^{\bm{\pi}} \leq x | {\bf X}_n) \overset{p}{\to} P(\xi \leq x)$ for all $x$ such that $P(\xi \leq x) $ is continuous, where $\xi \sim N(0_{(T-1)K \times 1}, \Sigma_Z)$. Then, the continuous mapping theorem from \citet[Lemma A.6]{chung/romano:2016} implies that for non-studentized statistic:
\begin{equation}
\label{eq:rand_dist_S}
    P(S_n^{\bm{\pi}} \leq x | {\bf X}_n) ~\overset{p}{\to}~ P(S \leq x),
\end{equation}
for all $x \in \mathbb{R}$, where $S$ is as in \eqref{eq:LimitingS}. This convergence relies on the fact that $\Sigma_Z \neq 0_{T(K-1)\times T(K-1)}$, which implies $P(S \leq x)$ is continuous for all $x \in \mathbb{R}$. From this point onward, the rest of the proof follows from arguments in part (a) of Theorem \ref{thm:AA_fixed}.

\noindent  \underline{Case 2: $\Sigma_Z = 0_{T(K-1)\times T(K-1)}$.} By the same arguments as in Theorem  \ref{thm:AA_fixed}, we have that $1(X_{i,t} \leq u_k) = 1(X_{i,t+1} \leq u_k)$  a.s.\ for all $t = 1,\ldots,T-1$ and $k = 1,\ldots K$, and so $S_n = 0$ a.s. Furthermore, for all $\pi \in \mathcal{M}$, we have that $1(X_{i,\pi(t)} \leq u_k) = 1(X_{i,\pi(t+1)} \leq u_k)$  a.s.\ for all $t = 1,\ldots,T-1$, $k = 1,\ldots K$. This implies that $S_n^{\bm{\pi}} = 0$ a.s., and so $S_n = c_n^\pi(1-\alpha) = 0$. The desired result follows from this and the construction of the test in \eqref{eq:test_perm}.

\noindent  \underline{Part (b).} We construct an example with $T=3$. For $\tau _{1},\tau _{2}\in ( 0,1) $, we focus on a Markov chain with two states, $s_1=0$ and $s_2=1$, and a transition matrix given by
\begin{equation*}
\left( 
\begin{array}{cc}
P( X_{i,t+1}=0|X_{i,t}=0) & P( X_{i,t+1}=1|X_{i,t}=0) \\
P( X_{i,t+1}=0|X_{i,t}=1) & P( X_{i,t+1}=1|X_{i,t}=1)
\end{array}
\right) ~=~\left( 
\begin{array}{cc}
\tau _{1} & 1-\tau _{1} \\ 
1-\tau _{2} & \tau _{2}
\end{array}
\right).
\end{equation*}
In the steady state, the marginal distribution is such that for all $t=1,2,3$,
\begin{equation}
 P( X_{i,t}=0)  ~=~ \frac{1-\tau _{2} }{2-\tau _{1}-\tau _{2}}~~~\text{and}~~~P( X_{i,t}=1)  ~=~ \frac{1-\tau _{1} }{2-\tau _{1}-\tau _{2}}. \label{eq:SS_markov}
\end{equation}

To assess the marginal homogeneity of this Markov chain on only two support points, it suffices to test the hypothesis at one of the two points (as the other is just its complement). For this reason, we construct our test statistic with $K=1$ and $u_1=0$.
It follows then
\[
    \Omega_Z~=~ 
    \frac{( 1-\tau _{1})(1-\tau_2)^2}{(2-\tau_1-\tau_2)^2 }\left[
    \begin{array}{cc}
    2  & -( 2-( \tau _{1}+\tau _{2}) ) \\
    -( 2-( \tau _{1}+\tau _{2}) )  & 2
    \end{array}
    \right].
\]

By $T = 3$, the $T! = 6$ permutations of $(1,2,3)$ are $ \mathcal{M}= \{(1,2,3),(1,3,2),(2,1,3),(2,3,1),(3,1,2),(3,2,1) \}$. Following the construction in step 2 of Theorem \ref{thm:perm_dist}, we have
\begin{align*}
    &B((1,2,3)) = \begin{bmatrix} 1 & 0 \\ 0 & 1 \end{bmatrix},~~~ B((1,3,2)) = \begin{bmatrix} 0 & -1 \\ -1 & 0 \end{bmatrix}, ~~~
    B((2,1,3)) = \begin{bmatrix} -1 & 0 \\ 0 &  1 \end{bmatrix}\\
    &B((2,3,1)) = \begin{bmatrix} 0 & 1 \\ -1 & -1\end{bmatrix}, ~~~
    B((3,1,2)) = \begin{bmatrix} -1 & -1 \\ 1 & 0\end{bmatrix},~~~
     B((3,2,1)) = \begin{bmatrix} 0 & -1 \\ -1 & 0\end{bmatrix},
\end{align*}
and
\[
    \Omega_Z^\pi ~=~ \tfrac{1}{6} \sum_{\pi \in \mathcal{M}} B(\pi) \Omega_Z B(\pi)'.
\]
It is not hard to verify that $\Omega_Z^\pi$ is PD and $\Omega_Z^\pi \neq \Omega_Z$. By the same arguments as in part (a), we conclude that
\begin{equation}
\label{eq:rand_dist_S2}
    P(S_n^{\bm{\pi}} \leq x | {\bf X}_n) ~\overset{p}{\to}~ P(S^{\pi} \leq x),
\end{equation}
for all $x \in \mathbb{R}$, where $S^{\pi} = \sum_{j=1}^{2}\lambda^{\pi} _{j}\zeta _{j}^{2}$ with$\{\zeta_j\}_{j=1}^{2}$ being i.i.d.\ $N(0,1)$, and $\{\lambda ^{\pi}_{j}\}_{j=1}^{2}$ are the eigenvalues of $\Omega_Z^\pi$. From this point onward, we can repeat arguments in part (a) of Theorem \ref{thm:AA_fixed} to show that 
\begin{align}
     c_n^\pi(1-\alpha)~\overset{p}{\to}~c^\chi(1-\alpha; \Omega_Z^\pi)~~~\text{and}~~~
     E_P[\phi^{\pi}_{n}(\alpha)]~{\to}~P[S \leq  c^\chi(1-\alpha; \Omega_Z^\pi)],\label{eq:distortion1}
\end{align}
where $S = \sum_{j=1}^{2}\lambda _{j}\zeta _{j}^{2}$ with $\{\lambda _{j}\}_{j=1}^{2}$ equal to the eigenvalues of $\Sigma_Z = \Omega_Z$. Since $\Omega_Z^\pi \neq \Omega_Z$, we have that $c^\chi(1-\alpha; \Omega_Z^\pi) \neq c^\chi(1-\alpha; \Omega_Z)$ and therefore $P[S \leq  c^\chi(1-\alpha; \Omega_Z^\pi)]\neq \alpha$. To show the asymptotic overrejection, it suffices to find examples of parameters in which $P[S \leq  c^\chi(1-\alpha; \Omega_Z^\pi)]>\alpha$. For instance, by choosing $\tau_1 = \tau_2=0.1$, we obtain 
\begin{align*}
     c^\chi(0.9; \Omega^\pi_Z) ~=~  1.56~&<~ c^\chi(0.9; \Sigma_Z) ~=~ 2.36,\\
 c^\chi(0.95; \Omega^\pi_Z) ~=~ 2.12~&<~ c^\chi(0.95; \Sigma_Z)~=~ 3.33,\\
 c^\chi(0.99; \Omega^\pi_Z) ~=~ 3.49~&<~ c^\chi(0.99; \Sigma_Z)~=~ 5.72,
\end{align*}
with the following asymptotic overrejection:
\begin{align*}
E_P[\phi^{\pi}_{n}( 0.1 ) ]\to 0.1828 > 0.1,~~~
E_P[\phi^{\pi}_{n}(0.05 ) ]\to 0.1194 > 0.05,~~~
E_P[\phi^{\pi}_{n}( 0.01 ) ] \to 0.0446 > 0.01.
\end{align*}

\noindent  \underline{Part (c).}  Under $H_0$, we have ${\Omega}_Z = \Sigma_Z$. Then, Assumption \ref{ass:NS} implies ${\Omega}_Z $ is PD. Thus, $B(\pi){\Omega}_Z B(\pi)'$ is PSD for all $\pi \in \mathcal{M}$, where $B(\pi)$ is defined in the proof of Theorem \ref{thm:perm_dist}. Furthermore, by choosing $\pi = (1,\dots,T)$, step 2 of Theorem \ref{thm:perm_dist} implies that $B((1,\dots,T)) = I_{(T-1)K \times (T-1)K}$, and so $B((1,\dots,T)){\Omega}_Z B((1,\dots,T))' = {\Omega}_Z$. Then, $\Omega_Z^{\pi} = \frac{1}{T!}\sum_{\pi \in \mathcal{M}} B(\pi) \Omega_Z B(\pi)'$ is PD, as it is the sum of PSD matrices with at least one PD matrix.

By Theorem \ref{thm:perm_dist}, Lemma \ref{lem:perm_cov}, and the fact that $\Omega_Z^{\pi}$ is PD, the Slutsky's theorem from \citet[Lemma A.5]{chung/romano:2016} implies that $P(\bar{S}_n^{\bm \pi} \leq x | {\bf X}_n) \overset{p}{\to } P( \chi_{(T-1)K}^2\leq x) $ for all $ x\in \mathbb{R} $. From this point onward, the rest of the proof follows from arguments in part (a) in Theorem \ref{thm:AA_fixed}.
\end{proof}


\begin{proof}[Proof of Theorem \ref{thm:fsperm}]
\noindent{\underline{Part (a)}}. Fix $( i,t,s) \in \{1,\ldots ,n\}\times \{1,\ldots ,T\}\times \{1,\ldots ,T\}$ arbitrarily, and let $\pi \in \mathcal{M}$ be any permutation that interchanges $t$ and $s$. For any $x\in \mathbb{R}$,
\begin{align*}
F_{t}( x)  &~=~P( X_{i,t}\leq x)  \\
&~=~\lim_{\{ u_{j}\} _{j\not=t}\to \infty }P( \{ X_{i,1},\ldots ,X_{i,t},\ldots ,X_{i,s},\ldots ,X_{i,T}\}  
\leq ( u_{1},\ldots ,x,\ldots ,u_{s},\ldots ,u_{T}) )  \\
&~\overset{(1)}{=}~\lim_{\{ u_{j}\} _{j\not=t}\to \infty }P( \{ X_{i,\pi ( 1) },\ldots ,X_{i,\pi ( t) },\ldots ,X_{i,\pi ( s) },\ldots ,X_{i,\pi ( T)}\}   
\leq ( u_{1},\ldots ,x,\ldots ,u_{s},\ldots ,u_{T}) )  \\
&~\overset{(2)}{=}~\lim_{\{ u_{j}\} _{j\not=t}\to \infty}P(  \{ X_{i,\pi ( 1) },\ldots ,X_{i,s},\ldots,X_{i,t},\ldots ,X_{i,\pi ( T) }\}  
\leq ( u_{1},\ldots ,x,\ldots ,u_{s},\ldots ,u_{T})  )  \\
&~=~P( X_{i,s}\leq x) =F_{s}( x) ,
\end{align*}
where (1) holds by $P\in \Omega _{{\rm TE}}$ and (2) by the specification of $\pi $. Since $x\in \mathbb{R}$ and $( t,s) \in \{1,\ldots ,T\}\times \{1,\ldots ,T\}$ were arbitrary, $H_{0}$ in \eqref{eq:H0} holds. 

To see that the reverse implication fails, consider the following example: $ \mathbf{X}_{n}=\{ ( X_{i,1},X_{i,2},X_{i,3}) \}_{i=1}^{n}$ i.i.d.\ with $X_{i,1}=X_{i,2}\perp X_{i,3}$ and $X_{i,t}\sim N( 0,1) $. It is not hard to verify that this distribution satisfies Assumption \ref{ass:A1} but does not belong to  $\Omega _{{\rm TE}}$.

\noindent{\underline{Part (b)}}.  Let ${\bf X}_n^\pi \equiv \{\{X_{i,\pi_i(t)}\}_{t=1}^T\}_{i=1}^n$ denote the sample permuted according to an arbitrary permutation $\pi^n = \{\pi_i\}_{i=1}^n \in \mathcal{M}^{n}$. Then,
\begin{equation}
        F_{{\bf X}_n} ~\overset{(1)}{=}~ \prod_{i=1}^n F_{X_{i,1},\ldots, X_{i,T}} ~\overset{(2)}{=}~ \prod_{i=1}^n F_{X_{i,\pi_i(1)},\ldots, X_{i,\pi_i(T)}} ~\overset{(3)}{=} F_{{\bf X}_n^\pi},
        \label{eq:randHyp}
\end{equation}
where (1) and (3) hold by Assumption \ref{ass:A1}, and (2) by $P\in \Omega _{{\rm TE}}$. 
We note that \eqref{eq:randHyp} implies that the randomization hypothesis (i.e., \citet[Definition 17.2.1]{lehmann/romano:2022}) holds. From here, \citet[Theorem 17.2.1]{lehmann/romano:2022} implies that the permutation test described in \citet[Section 17.2.1]{lehmann/romano:2022} satisfies \eqref{eq:FS_SizeControl} with equality. In turn, this implies that our permutation test (i.e., the non-random version of the test in \citet[Section 17.2.1]{lehmann/romano:2022}) satisfies \eqref{eq:FS_SizeControl}.
\end{proof}


\begin{lemma} \label{lem:perm_cov}
Under Assumption \ref{ass:A1},
\[
\hat{\Sigma}_{Z^{\bm{\pi} }}~\overset{p}{\to}~\Omega _{Z}^{\pi }.
\]
\end{lemma}
\begin{proof}
This proof relies on notation and arguments in the proof of Theorem \ref{thm:perm_dist}. First, we show that
\begin{equation}
\frac{1}{n}\sum_{i=1}^{n}B( \bm{\pi }_{i}) Z_{i}Z_{i}^{\prime }B( \bm{\pi }_{i}) ^{\prime }~\overset{p}{\to}~\Omega _{Z}^{\pi }. \label{eq:perm_cov1}
\end{equation}
By similar arguments as in step 3 of Theorem \ref{thm:perm_dist}, we note that $\{ B( \bm{\pi }_{i}) Z_{i}Z_{i}^{\prime }B( \bm{\pi } _{i}) ^{\prime }\} _{i=1}^{n}$ is i.i.d.\ with
\[
E[ B( \bm{\pi }_{i}) Z_{i}Z_{i}^{\prime }B( \bm{ \pi }_{i}) ^{\prime }] ~\overset{(1)}{=}~E[ B( \bm{\pi } _{i}) \Omega _{Z}B( \bm{\pi }_{i}) ^{\prime }] ~\overset{(2)}{=}~\Omega _{Z}^{\pi }.
\]
where (1) holds by $\{ B( \bm{\pi }_{i}) \} _{i=1}^{n}\perp \{ Z_{i}\} _{i=1}^{n}$ and $E[ Z_{i}Z_{i}^{\prime }] =\Omega _{Z}$, and (2) by \eqref{eq:Omega_pi} and the fact that $\bm{\pi }_{i}$ is uniformly distributed in $\mathcal{M}$. From these observations and the LLN, \eqref{eq:perm_cov1} follows.

Second, we note that
\begin{equation}
\frac{1}{n}\sum_{i=1}^{n}B( \bm{\pi }_{i}) Z_{i}V_{i}^{\prime }B( \bm{\pi }_{i}) ^{\prime }=O_{p}( 1) ,~~~ \frac{1}{n}\sum_{i=1}^{n}B( \bm{\pi }_{i}) V_{i}Z_{i}^{\prime }B( \bm{\pi }_{i}) ^{\prime }=O_{p}( 1)\text{, and }~~~\frac{1}{n}\sum_{i=1}^{n}B( \bm{\pi }_{i}) V_{i}V_{i}^{\prime }B( \bm{\pi }_{i}) ^{\prime }=O_{p}( 1) . \label{eq:perm_cov2}
\end{equation}
These can be shown by the arguments that yield \eqref{eq:perm_cov1}, except that $Z_{i}Z_{i}^{\prime }$ is replaced by $Z_{i}V_{i}^{\prime }$, $V_{i}Z_{i}^{\prime }$, and $ V_{i}V_{i}^{\prime }$, respectively.

To conclude the proof, consider the following derivation.
\begin{align*}
\hat{\Sigma}_{Z^{\bm{\pi} }}
& ~\overset{(1)}{=}~\frac{1}{n}\sum_{i=1}^{n}\hat{M}V_{i}^{\bm{\pi} }( \hat{M}V_{i}^{\bm{\pi} }) ^{\prime }-\left( \frac{1}{n}\sum_{i=1}^{n}\hat{M}V_{i}^{\bm{\pi} }\right) \left( \frac{1}{n}\sum_{i=1}^{n}\hat{M}V_{i}^{\bm{\pi} }\right) ^{\prime } \\
& ~\overset{(2)}{=}~\hat{M}\left( \frac{1}{n}\sum_{i=1}^{n}B( \bm{ \pi }_{i}) V_{i}V_{i}^{\prime }B( \bm{\pi }_{i}) ^{\prime }\right) \hat{M}^{\prime }-\hat{M}\left( \frac{1}{n} \sum_{i=1}^{n}B( \bm{\pi }_{i}) V_{i}\right) \left( \frac{1}{n }\sum_{i=1}^{n}B( \bm{\pi }_{i}) V_{i}\right) ^{\prime }\hat{M }^{\prime } \\
& ~\overset{(3)}{=}~( M+o_{p}( 1) ) \left( \frac{1}{n} \sum_{i=1}^{n}B( \bm{\pi }_{i}) V_{i}V_{i}^{\prime }B( \bm{\pi }_{i}) ^{\prime }\right) ( M+o_{p}( 1) ) ^{\prime }+o_{p}( 1) \\
& ~\overset{(4)}{=}~\left[
\begin{array}{c}
\frac{1}{n}\sum_{i=1}^{n}B( \bm{\pi }_{i}) Z_{i}Z_{i}^{\prime }B( \bm{\pi }_{i}) ^{\prime }+( \frac{1}{n} \sum_{i=1}^{n}B( \bm{\pi }_{i}) Z_{i}V_{i}^{\prime }B( \bm{\pi }_{i}) ^{\prime }) o_{p}( 1) + \\
o_{p}(1) ( \frac{1}{n}\sum_{i=1}^{n}B( \bm{\pi } _{i}) V_{i}V_{i}^{\prime }B( \bm{\pi }_{i}) ^{\prime }) o_p(1) +o_{p}( 1) ( \frac{1}{n}\sum_{i=1}^{n}B( \bm{\pi }_{i}) V_{i}Z_{i}^{\prime }B( \bm{\pi } _{i}) ^{\prime })
\end{array}
\right] +o_{p}( 1) \\
&~ \overset{(5)}{=}~\Omega _{Z}^{\pi }+o_{p}( 1) ,
\end{align*}
as desired, where (1) holds by \eqref{eq:M}, \eqref{eq:perm_defn}, and \eqref{eq:P_pi_constant}, (2) by \eqref{eq:B_perm}, (3) by \eqref{eq:Mconvergence} and $\frac{1 }{n}\sum_{i=1}^{n}B( \bm{\pi }_{i}) V_{i}=o_{p}( 1) $ (implied by \eqref{eq:OpWithIV}), (4) by $M B( \bm{\pi }_{i})V_{i} = B( \bm{\pi }_{i})Z_{i}$ for all $i=1,\dots,n$ (implied by $Z_{i}=MV_{i}$ for all $i=1,\dots,n$ and \eqref{eq:interchange}), and (5) by \eqref{eq:perm_cov1} and \eqref{eq:perm_cov2}.
\end{proof}

\begin{lemma}\label{lem:Pauly} 
Let $\Omega _{{\rm FE}}$ denote the class of distributions that are ``fully'' exchangeable over units and time periods, i.e., $\mathbf{X}_{n}=\{ \{ X_{i,t}\} _{i=1}^{n}\} _{t=1}^{T}$ has the same distribution as $\{ \{ X_{\lambda ( i,t)}\} _{i=1}^{n}\} _{t=1}^{T}$ where $\lambda$ denotes an arbitrary permutation of units and time periods. Under Assumption \ref{ass:A1}, the following statements hold:
\begin{enumerate}[(a)]
\item If $n=1$, $ \Omega _{{\rm FE}}=\Omega _{{\rm TE}}$.
\item If $n>1$, $\Omega _{{\rm FE}}\subsetneq \Omega _{{\rm TE}}$.
\end{enumerate}
\end{lemma}
\begin{proof} Part (a) is straightforward, so we prove part (b).
We begin by showing a useful intermediate result: $\mathbf{X}_{n} \sim P \in \Omega _{{\rm FE}}$ implies that its CDF can be written as
\begin{equation}
P( \mathbf{X}_{n}\leq y)~=~\prod\limits_{i=1}^{n}\prod\limits_{t=1}^{T}F( y_{i,t}) ~~~\text{for all }y\in \mathbb{R} ^{nT},
\label{eq:pauly1}
\end{equation}
where $F$ is the CDF of $X_{1,1}$. Since Assumption \ref{ass:A1} already implies independence across units:
\begin{align}
P( \mathbf{X}_{n}\leq y) ~=~\prod\limits_{i=1}^{n}\tilde{F}( \{ y_{i,t}\} _{t=1}^{T}) ,
\label{eq:pauly1p5}
\end{align}
where $\tilde{F}$ is the CDF of the vector $\{ X_{1,t}\} _{t=1}^{T}$. Then the desired result \eqref{eq:pauly1} follows immediately from \eqref{eq:pauly1p5} provided that $\{ X_{1,t}\} _{t=1}^{T}$ are i.i.d.\ with marginal CDF $F$. We now establish this result in two steps.

First, we show $\{ X_{1,t}\} _{t=1}^{T}$ is an independent sequence. To this end, fix $\{ x_{t}\} _{t=1}^{T}\in \mathbb{R}^{T}$ arbitrarily. For any $s=1,\ldots ,T-1$, consider the following permutation: $\lambda ( 1,t) =( 1,t) $ for $t\leq s$ and $\lambda ( 1,t) =( 2,t) $ for $t>s$, $\lambda ( 2,t) =( 2,t) $ for $t\leq s$ and $\lambda ( 2,t) =( 1,t) $ for $t>s$, and $\lambda ( j,t) =( j,t) $ for all $j>2$ and $t=1,\ldots ,T$. Then,
\begin{align}
&P( X_{1,1}\leq x_{1},\ldots ,X_{1,s}\leq x_{s},X_{1,s+1}\leq x_{s+1},\ldots ,X_{1,T}\leq x_{T}) \notag\\
& =\lim_{u\rightarrow \infty }P\left( 
\begin{array}{c}
X_{1,1}\leq x_{1},\ldots ,X_{1,s}\leq x_{s},X_{1,s+1}\leq x_{s+1},\ldots ,X_{1,T}\leq x_{T} \\ 
X_{j,t}\leq u\text{ for all }j\geq 2\text{ and }t=1,\ldots ,T
\end{array}
\right)   \notag\\
& \overset{(1)}{=}\lim_{u\rightarrow \infty }P\left( 
\begin{array}{c}
X_{\lambda ( 1,1)}\leq x_{1},\ldots ,X_{\lambda ( 1,s)}\leq x_{s},X_{\lambda ( 1,s+1)}\leq x_{s+1},\ldots ,X_{\lambda ( 1,T)}\leq x_{T} \\ 
X_{\lambda ( j,t)}\leq u\text{ for all }j\geq 2\text{ and }t=1,\ldots ,T
\end{array}
\right)  \notag\\
& \overset{(2)}{=}\lim_{u\rightarrow \infty }P\left( 
\begin{array}{c}
X_{1,1}\leq x_{1},\ldots ,X_{1,s}\leq x_{s},X_{2,s+1}\leq x_{s+1},\ldots ,X_{2,T}\leq x_{T} \\ 
X_{\lambda ( j,t)}\leq u\text{ for all }j\geq 2\text{ and }t=1,\ldots ,T
\end{array}
\right)  \notag\\
& =P( X_{1,1}\leq x_{1},\ldots ,X_{1,s}\leq x_{s},X_{2,s+1}\leq x_{s+1},\ldots ,X_{2,T}\leq x_{T})   \notag\\
& \overset{(3)}{=}P( X_{1,1}\leq x_{1},\ldots ,X_{1,s}\leq x_{s}) P( X_{2,s+1}\leq x_{s+1},\ldots ,X_{2,T}\leq x_{T}) ,\label{eq:pauly2}
\end{align}
where (1) holds by $P\in \Omega _{{\rm FE}}$, (2) by the specification of $\lambda $, and (3) by Assumption \ref{ass:A1}. By taking limits of \eqref{eq:pauly2} as $x_{1},\ldots ,x_{s}\to \infty$, we get 
\begin{equation}
P( X_{1,s+1}\leq x_{s+1},\ldots ,X_{1,T}\leq x_{T}) =P( X_{2,s+1}\leq x_{s+1},\ldots ,X_{2,T}\leq x_{T}) .
\label{eq:pauly3}
\end{equation}
Since \eqref{eq:pauly3} holds for all $\{ x_{t}\} _{t=1}^{T}\in \mathbb{R}^{T}$, we can combine it with \eqref{eq:pauly2} to get
\begin{equation}
( X_{1,1},\ldots ,X_{1,s}) \perp ( X_{1,s+1},\ldots ,X_{1,T}) \text{ for any }s=1,\ldots ,T-1.
\label{eq:pauly4}
\end{equation}
The desired result follows by considering \eqref{eq:pauly4} sequentially for $s=1$, $s=2$, and so on.

Second, we show that $\{ X_{1,t}\} _{t=1}^{T}$ is an identically distributed sequence. To this end, fix $x\in \mathbb{R}$ arbitrarily. For any $s\neq 1$, consider the following permutation: $\lambda ( 1,1) =( 1,s) $, $\lambda ( 1,s) =\lambda ( 1,1) $, and $\lambda ( i,t) =( i,t) $ otherwise. Then,
\begin{align}
P( X_{1,1}\leq x) & =\lim_{u\rightarrow \infty }P( X_{1,1}\leq x,X_{i,t}\leq u\text{ for all }( i,t) \not=( 1,1)) \notag\\
& \overset{(1)}{=}\lim_{u\rightarrow \infty }P( X_{\lambda ( 1,1)}\leq x,X_{\lambda ( i,t)}\leq u\text{ for all }( i,t) \not=( 1,1)) \notag\\
& \overset{(2)}{=}\lim_{u\rightarrow \infty }P( X_{1,s}\leq x,X_{\lambda(i,t)}\leq u\text{ for all }( i,t) \not=( 1,1)) \notag\\
& =P( X_{1,s}\leq x) ,\label{eq:pauly5}
\end{align}
where (1) holds by $P\in \Omega _{{\rm FE}}$ and (2) by the specification of $\lambda $. Since \eqref{eq:pauly5} holds for all $x\in \mathbb{R}$, $X_{1,1}$ and $X_{1,s}$ have the same distribution. Since the choice of $s=2,\dots,T$ was arbitrary, the desired result follows.

Finally, we conclude the proof by finding a distribution $P\in\Omega_{\text{TE}}$ but $P\not\in \Omega_{\text{FE}}$, so the inclusion is strict. Consider the following example: $\mathbf{X}_{n}=\{ \{ X_{i,t}\} _{i=1}^{n}\} _{t=1}^{T}$ with $n=2$, $T=2$, where $X_{1,1}=X_{1,2}=Z_{1}$, $X_{2,1}=X_{2,2}=Z_{2}$, and $\{ Z_{1},Z_{2}\}$ are i.i.d.\ $N( 0,1) $. It is easy to see that this distribution satisfies Assumption \ref{ass:A1} and $P\in \Omega _{{\rm TE}}$, however $P\not\in \Omega _{{\rm FE}}$.
\end{proof}